\documentclass[11pt,a4paper]{article}

\usepackage{hyperref}
\usepackage{cite}

\usepackage{amsmath}
\usepackage{amsfonts,amsthm,amssymb}

\numberwithin{equation}{section}

\newcommand{\beqa}{\begin{eqnarray}}
\newcommand{\eeqa}{\end{eqnarray}}

\newtheorem{theorem}{Theorem}[section]
\newtheorem{proposition}{Proposition}[section]
\newtheorem{lemma}{Lemma}[section]

\newtheorem{corollary}{Corollary}[section]

{\theoremstyle{remark}
}
\newtheorem{identity}{Identity}
\numberwithin{equation}{section}

\begin{document}

\begin{flushright}
LPENSL-TH-11-18
\end{flushright}

\bigskip \bigskip

\begin{center}
\textbf{{\LARGE Complete spectrum of quantum integrable lattice models associated to $\mathcal{U}_{q} (\widehat{gl_{n}})$ by separation of variables}}

\vspace{50pt}

{\large \textbf{J.~M.~Maillet}\footnote[1]{{ {Univ Lyon, Ens de Lyon,
Univ Claude Bernard Lyon 1, CNRS, Laboratoire de Physique, UMR 5672, F-69342
Lyon, France; maillet@ens-lyon.fr}}} \textbf{\, and} ~~ \textbf{G. Niccoli}%
\footnote[2]{{ {Univ Lyon, Ens de Lyon, Univ Claude Bernard Lyon 1,
CNRS, Laboratoire de Physique, UMR 5672, F-69342 Lyon, France;
giuliano.niccoli@ens-lyon.fr}}}}
\end{center}

\vspace{40pt}

\begin{itemize}
\item[ ] \textbf{Abstract.}\thinspace\ In this paper we apply our new
separation of variables approach to completely characterize the transfer
matrix spectrum for quantum integrable lattice models associated to fundamental evaluation representations of $\mathcal{U}_{q} (\widehat{gl_{n}})$ with general quasi-periodic
boundary conditions. We consider here the case of generic deformations associated to a parameter $q$ which is not a root of unity. The Separation of Variables (SoV) basis for the transfer matrix spectral problem is generated by using the action of the transfer matrix itself on a generic co-vector of the Hilbert space,  following the general procedure described in our paper \cite{MaiN18}. Such a SoV construction allows to prove that for general values of
the parameters defining the model the transfer matrix is diagonalizable and with simple
spectrum for any twist matrix which is also diagonalizable with simple spectrum. Then, using together the knowledge of such a SoV basis and of  the fusion relations satisfied by the hierarchy of transfer matrices, we derive a complete characterization
of the transfer matrix eigenvalues and eigenvectors as solutions of a system of polynomial equations of order $n+1$. Moreover, we show that such a SoV discrete spectrum
characterization is equivalently reformulated in terms of a finite
difference functional equation, the quantum spectral curve equation, under
a proper choice of the set of its solutions. A construction of the associated $Q$%
-operator induced by our SoV approach is also presented.
\end{itemize}

\newpage
\tableofcontents
 \newpage

\section{Introduction}

In this paper we make use of our new approach to generate the SoV basis \cite%
{MaiN18,MaiN18a,MaiN18b} to characterize the complete spectrum of quantum integrable lattice models with general quasi-periodic
boundary conditions associated to the fundamental evaluation representations of
higher rank trigonometric Yang-Baxter algebras. More in detail, in the framework of the quantum inverse
scattering method \cite%
{FadS78,FadST79,FadT79,Skl79,Skl79a,FadT81,Skl82,Fad82,Fad96}, these models
are constructed by using the rank $n$ principal gradation $R$-matrix \cite%
{BabdeVV81,Jim86} solution of the Yang-Baxter equation associated to the
quantum group $\mathcal{U}_{q}(\widehat{gl_{n+1}})$ \cite%
{KulR83a,KulR82,Jim85,Dri87,ChaP94}. This $R$-matrix admits a nontrivial set
of scalar solutions to the Yang-Baxter equation. Such symmetries of the $R$%
-matrix, here called twist $K$-matrices, allow for the definition of
integrable quasi-periodic boundary conditions for the corresponding integrable quantum
models. Our SoV basis is generated for these
general twists, under the assumption that the twist matrix $K$ has simple
spectrum. As previously shown in our work \cite%
{MaiN18}, the existence of such a SoV basis implies the simplicity of the transfer matrix spectrum. Moreover if the twist matrix $K$ is diagonalizable with simple spectrum it implies that the transfer matrix is also diagonalizable with simple spectrum for almost any choice of the parameters of the representations.

The transfer matrix spectrum of these integrable quantum models has been
analyzed also by other approaches, in particular, for diagonal twists, in the
framework of the fusion relations \cite{KulRS81,KulR83} and analytic Bethe
ansatz \cite{Res83,Res83a,KirR86,Res87,KunNS94}, the nested Bethe ansatz 
\cite{KulR81, KulR83,BelR08,PakRS18}, with also first
interesting results toward the computation of correlations functions \cite%
{BelPRS12,BelPRS12a,BelPRS13,BelPRS13a,PakRS14,PakRS14a,PakRS15,PakRS15a,PakRS15b,LiaS18}%
. Let us also note that for anti-periodic boundary conditions an eigenvalue analysis by a functional
approach has been developed in \cite{HaoCLYSW16}.

The quantum version of the separation of variables method has been pioneered by Sklyanin in  a series of works
\cite{Skl85,Skl90,Skl92,Skl92a,Skl95,Skl96} in particular to tackle models for which the standard algebraic Bethe ansatz cannot be applied. Since the Sklyanin's original
papers this approach has been successfully implemented and partially
generalized to several classes of integrable quantum models mainly
associated to different representations of quantum algebras related of rank one type, e.g. for the 6-vertex
and 8-vertex Yang-Baxter algebras and reflection algebras as well as to
their dynamical deformations \cite%
{Skl85,Skl90,Skl92,Skl92a,Skl95,Skl96,BabBS96,Smi98a,Smi01,DerKM01,DerKM03,DerKM03b,BytT06,vonGIPS06,FraSW08,MukTV09,MukTV09a,MukTV09c,AmiFOW10,NicT10,Nic10a,Nic11,FraGSW11,GroMN12,GroN12,Nic12,Nic13,Nic13a,Nic13b,GroMN14,FalN14,FalKN14,KitMN14,NicT15,LevNT15,NicT16,KitMNT16,MarS16,KitMNT17,MaiNP17}%
. The interest in developing the separation of variables method is mainly
due to some important built-in features as the ability to provide a direct 
proof of the completeness of the spectrum characterization as well as some
first elements towards the dynamics like scalar products and form factors.

Important analysis toward the SoV description of higher rank cases have been presented in \cite{Skl96,Smi01,MarS16,GroLMS17,DerV18}.
Here, we solve the long-standing problem to systematically introduce a
quantum separation of variable approach capable to completely characterize
the transfer matrix spectrum associated to the higher rank representations
of the trigonometric Yang-Baxter algebra. While our approach bypass the
construction of the so-called Sklyanin's commuting $B$-operator family \cite%
{Skl85,Skl90,Skl92,Skl92a,Skl95,Skl96}, the results on the rank one
representations as well as some evidence from the short lattices for the
higher rank representations \cite{MaiN18}, plus some recent analysis
developed in \cite{RyaV18} for the rational higher rank situation, confirm that our SoV basis
construction can nevertheless reproduce the Sklyanin's SoV basis (i.e. the $B$-operator
eigenbasis). This is the case under some special choice of the generating
covector, i.e. the covector from which our SoV basis is constructed by the
action of a chosen set of commuting conserved charges. This
type of connection for the trigonometric representations, in particular,
in relation to the SoV results obtained in \cite{Smi01} would be interesting to study further.

However, as anticipated above, our strategy following \cite{MaiN18}, is instead to make a direct use of our SoV basis construction to obtain the complete characterization of the spectrum (eigenvalues and eigenvectors) using in particular the hierarchy of fusion relations for the transfer matrices. Here we  consider the transfer matrices associated to general twist matrices $K$ diagonalizable and with simple spectrum for the trigonometric $gl_n$ ($n\geq 2$ ) Yang-Baxter algebras in the fundamental evaluation 
representations. We will first obtain a complete characterization of the spectrum in terms of the set of
solutions to a given system of $\mathsf{N}$ polynomial equations of degree $%
n+1$ in $\mathsf{N}$ unknowns, where $\mathsf{N}$ is the number of lattice sites. Second, we introduce the so-called quantum spectral curve functional equation and we provide the exact characterization of the set of its
solutions which generates the complete transfer matrix spectrum associating
to any eigenvalue solution exactly one nonzero eigenvector. These results allow also to point
out, as already done for other quantum integrable models \cite{MaiN18,MaiN18a,MaiN18b}, how the SoV basis in our construction can be equivalently obtained by the action of the commuting family of $Q$-operator. This connection is important as it allows to bring in our SoV approach results of the  Baxter's Q-operator method \cite{Bax73-tot,Bax73,Bax76,Bax77,Bax78,PasG92,BatBOY95,YunB95,BazLZ97,AntF97,BazLZ99,Der99,Pro00,Kor06,Der08,BazLMS10,Man14} ; of special interests are then the results presented in \cite{MT15} for the higher rank case. 
In our approach we show that this $Q$-operator satisfies with the transfer matrices the quantum spectral curve equation and that it can be reconstructed, making use of our Sov basis, in terms of the monodromy matrix entries. 

In order to make easier the reading of our
results we have decided to present them first in the rank 2 case, namely for the
fundamental representations generated by the principal gradation $R$-matrix
associated to $\mathcal{U}_{q}(\widehat{gl_{3}})$. Then these results are extended to the generic higher rank cases associated to the fundamental evaluation representations of $\mathcal{U}_{q}(\widehat{gl_{n}})$. 

This article is organized as follows. In section
2, we introduce the fundamental evaluation representation of the rank two trigonometric Yang-Baxter
algebra, the corresponding quantum spectral invariants of the model and then we list some
general analytic properties they satisfy. Then in subsection 2.2 we
construct our SoV covector basis and we state the first consequences on the
transfer matrix spectrum. The section 3 is then dedicated to the
presentation of our results on the transfer matrix spectrum
characterization. In subsection 3.1, we derive the SoV discrete
characterization of the transfer matrix spectrum in terms of solutions to a
system of polynomial equations of degree 3. In subsection 3.2, we give an equivalent characterization in terms of the solutions to a
functional equation of third order type, the so-called quantum spectral
curve equation. In subsection 3.3 we also show that our SoV characterization of the
transfer matrix eigenvectors allows for their rewriting in an algebraic Bethe
ansatz form. In section 4, we define the framework of the general higher rank $n$
case by introducing the corresponding fundamental evaluation representations of the $\mathcal{U}_{q}(\widehat{gl_{n}})$  Yang-Baxter
algebra, the associated quantum spectral invariants and some of their general analytic properties. In subsection 4.2 we construct our SoV covector basis for
this general rank $n$ case. The section 5 presents the complete 
transfer matrix spectrum characterization. We first derive the SoV discrete
spectrum characterization in subsection 5.1 while in subsection 5.2 we show
its equivalence to the quantum spectral curve, a functional equation of
difference type of order $n+1$ for the $Q$-operator that we construct using the knowledge of our SoV basis. Some important technical proofs are gathered in two appendices. In appendix A, for the rank
two case, we provide a proof of the complete characterization of the spectrum which is based on the explicit
calculation of the transfer matrix action on our SoV covector basis. While this proof can be generalized to the general rank $n$
along a similar path described in \cite{MaiN18a} for the rational case, we provide in appendix B a proof of the SoV discrete characterization of the transfer matrix spectrum which bypass the computation of the action of the
transfer matrix on the SoV covector basis. 

 \newpage
\section{Transfer matrices for fundamental evaluation representations of $\mathcal{U}_{q}(\widehat{gl_{3}})$}

We consider here the trigonometric Yang-Baxter algebra generated by the principal
gradation $R$-matrix \cite{BabdeVV81,Jim86} associated to the quantum group $\mathcal{U}_{q}(\widehat{gl_{3}})$\footnote{As already stressed, in all this article we assume that the deformation parameter $q=e^{\eta}$ is not a root of unity}:%
\begin{equation}
R_{a,b}^{\left( P\right) }(\lambda )=\left( 
\begin{array}{ccc}
a_{1}(\lambda ) & \lambda ^{1/3}b_{1} & \lambda ^{-1/3}b_{2} \\ 
\lambda ^{-1/3}c_{1} & a_{2}(\lambda ) & \lambda ^{1/3}b_{3} \\ 
\lambda ^{1/3}c_{2} & \lambda ^{-1/3}c_{3} & a_{3}(\lambda )%
\end{array}%
\right) \in \text{End}(V_{a}\otimes V_{b}),
\end{equation}%
where $V_{a}\cong V_{b}\cong \mathbb{C}^{3}$ and we have defined:%
\begin{align}
& a_{j}(\lambda )\left. =\right. \left( 
\begin{array}{ccc}
\lambda q^{\delta _{j,1}}-1/(\lambda q^{\delta _{j,1}}) & 0 & 0 \\ 
0 & \lambda q^{\delta _{j,2}}-1/(\lambda q^{\delta _{j,2}}) & 0 \\ 
0 & 0 & \lambda q^{\delta _{j,3}}-1/(\lambda q^{\delta _{j,3}})%
\end{array}%
\right) ,\text{ \ \ }\forall j\in \{1,2,3\},  \notag \\
& b_{1}\left. =\right. \left( 
\begin{array}{ccc}
0 & 0 & 0 \\ 
q-1/q & 0 & 0 \\ 
0 & 0 & 0%
\end{array}%
\right) ,\text{ \ }b_{2}\left. =\right. \left( 
\begin{array}{ccc}
0 & 0 & 0 \\ 
0 & 0 & 0 \\ 
q-1/q & 0 & 0%
\end{array}%
\right) ,\text{ \ }b_{3}\left. =\right. \left( 
\begin{array}{ccc}
0 & 0 & 0 \\ 
0 & 0 & 0 \\ 
0 & q-1/q & 0%
\end{array}%
\right) ,  \notag \\
& c_{1}\left. =\right. \left( 
\begin{array}{ccc}
0 & q-1/q & 0 \\ 
0 & 0 & 0 \\ 
0 & 0 & 0%
\end{array}%
\right) ,\text{ \ }c_{2}\left. =\right. \left( 
\begin{array}{ccc}
0 & 0 & q-1/q \\ 
0 & 0 & 0 \\ 
0 & 0 & 0%
\end{array}%
\right) ,\text{ \ }c_{3}\left. =\right. \left( 
\begin{array}{ccc}
0 & 0 & 0 \\ 
0 & 0 & q-1/q \\ 
0 & 0 & 0%
\end{array}%
\right) .
\end{align}%
Note that here we have chosen to present the $R$-matrix in a Laurent polynomial form but
clearly it can be rewritten as well in a trigonometric form. This $R$-matrix is a solution of the Yang-Baxter
equation:%
\begin{equation}
R_{12}(\lambda /\mu )R_{13}(\lambda )R_{23}(\mu )=R_{23}(\mu )R_{13}(\lambda
)R_{12}(\lambda /\mu )\in \text{End}(V_{1}\otimes V_{2}\otimes V_{3}),
\end{equation}%
and it is related by a similarity transformation to the so-called
homogeneous gradation $R$-matrix for $U_{q}(\widehat{gl_{3}})$:%
\begin{equation}
R_{a,b}^{(H)}(\lambda )=\left( 
\begin{array}{ccc}
a_{1}(\lambda ) & \lambda b_{1} & \lambda b_{2} \\ 
c_{1}/\lambda & a_{2}(\lambda ) & \lambda b_{3} \\ 
c_{2}/\lambda & c_{3}/\lambda & a_{3}(\lambda )%
\end{array}%
\right) \in \text{End}(V_{a}\otimes V_{b}),
\end{equation}%
which is also a solution of the Yang-Baxter equation. More in detail, it
holds:%
\begin{equation}
R_{a,b}^{\left( P\right) }(\lambda )=S_{a}(\lambda )R_{a,b}^{(H)}(\lambda
)S_{a}^{-1}(\lambda ),
\end{equation}%
where%
\begin{equation}
S(\lambda )=\left( 
\begin{array}{ccc}
1 & 0 & 0 \\ 
0 & \lambda ^{2/3} & 0 \\ 
0 & 0 & \lambda ^{4/3}%
\end{array}%
\right) ,
\end{equation}%
such connection has been first remarked in \cite{Jim86}. Let us comment that
this two solutions of the Yang-Baxter equation generate the same quantum
integrable models with diagonal quasi-periodic boundary conditions. Indeed,
any diagonal $3\times 3$ matrix $K\in $End$(V)$ is a symmetry for both these 
$R$-matrices:%
\begin{equation}
R_{12}^{\left( P/H\right) }(\lambda )K_{1}K_{2}=K_{2}K_{1}R_{12}^{\left(
P/H\right) }(\lambda )\in \text{End}(V_{1}\otimes V_{2}),
\end{equation}%
and defined the monodromy matrices:%
\begin{eqnarray}
M_{a}^{(P/H|K)}(\lambda |\{\xi \}) &=&\left( 
\begin{array}{ccc}
A_{1}^{(K)}(\lambda ) & B_{1}^{(K)}(\lambda ) & B_{2}^{(K)}(\lambda ) \\ 
C_{1}^{(K)}(\lambda ) & A_{2}^{(K)}(\lambda ) & B_{3}^{(K)}(\lambda ) \\ 
C_{2}^{(K)}(\lambda ) & C_{3}^{(K)}(\lambda ) & A_{3}^{(K)}(\lambda )%
\end{array}%
\right) _{a} \\
&\equiv &K_{a}R_{a,\mathsf{N}}^{\left( P/H\right) }(\lambda /\xi _{\mathsf{N}%
})\cdots R_{a,1}^{\left( P/H\right) }(\lambda /\xi _{1})\in \text{End}%
(V_{a}\otimes \mathcal{H}),
\end{eqnarray}%
where $\mathcal{H}=\bigotimes_{n=1}^{\mathsf{N}}V_{n}$, and the transfer
matrices%
\begin{equation}
T_{1}^{\left( P/H|K\right) }(\lambda |\{\xi \})\equiv \text{tr}%
_{a}M_{a}^{\left( P/H|K\right) }(\lambda |\{\xi \})\in \text{End}\mathcal{H}
\end{equation}%
then the following identity holds for the homogeneous chains:%
\begin{equation}
\left. T_{1}^{\left( P|K\right) }(\lambda |\{\xi \})\right\vert _{\xi
_{i}=0}=\left. T_{1}^{\left( H|K\right) }(\lambda |\{\xi \})\right\vert
_{\xi _{i}=0},
\end{equation}%
which implies our statement on the equality of the Hamiltonians with diagonal
quasi-periodic boundary conditions generated by these two $R$-matrices.
However, one has to remark that for the inhomogeneous models the above
identity does not hold.

One motivation to consider the fundamental models associated with the
principal gradation $R$-matrix is the set of symmetry $K$-matrices
enjoyed by this Yang-Baxter solution together with the simpler properties under co-product action, and hence under fusion. Indeed, the set of solutions of the
scalar Yang-Baxter equation for the principal gradation $R$-matrix reads:$%
\allowbreak $%
\begin{equation}
K^{\left( a\right) }=\delta _{a,1}\left( 
\begin{array}{ccc}
\alpha & 0 & 0 \\ 
0 & \beta & 0 \\ 
0 & 0 & \gamma%
\end{array}%
\right) +\delta _{a,2}\left( 
\begin{array}{ccc}
0 & 0 & \alpha \\ 
\beta & 0 & 0 \\ 
0 & \gamma & 0%
\end{array}%
\right) +\delta _{a,3}\left( 
\begin{array}{ccc}
0 & \beta & 0 \\ 
0 & 0 & \gamma \\ 
\alpha & 0 & 0%
\end{array}%
\right) ,
\end{equation}%
for any complex value of $\alpha ,$ $\beta $\ and $\gamma $, while for the
homogeneous gradation $R$-matrix the matrices $K^{\left( 2\right) }$ and $%
K^{\left( 3\right) }$ are symmetries if and only if $\alpha =0$. Note that
the matrices $K^{\left( 2\right) }$ and $K^{\left( 3\right) }$ are
diagonalizable and their eigenvalues reads:%
\begin{equation}
\mathsf{k}_{0}=\sqrt[3]{\alpha \beta \gamma },\mathsf{k}_{1}=-(-1)^{1/3}%
\mathsf{k}_{0},\mathsf{k}_{2}=(-1)^{2/3}\mathsf{k}_{0}
\end{equation}%
so that for $\alpha \beta \gamma \neq 0$ these matrices have simple spectrum and it
holds:%
\begin{equation*}
\text{det}K^{\left( a\right) }=\alpha \beta \gamma \text{ \ }\forall a\in
\{1,2,3\}.
\end{equation*}%
It is also interesting to remark that the principal gradation $R$%
-matrix directly generates, the $q$-deformed antisymmetric projector $%
P_{a,b}^{-}\in  $End$ (V_{a}\otimes V_{b})$ which is used in the fusion representations  for
the $\mathcal{U}_{q}(\widehat{gl_{3}})$ case. In particular, it holds:%
\begin{align}
& \frac{R_{12}^{\left( P\right) }(1/q)}{2(1/q-q)}=P_{a,b}^{-}  \notag \\
& =\left( 
\begin{array}{ccccccccc}
0 & 0 & 0 & 0 & 0 & 0 & 0 & 0 & 0 \\ 
0 & 1/2 & 0 & -q^{1/3}/2 & 0 & 0 & 0 & 0 & 0 \\ 
0 & 0 & 1/2 & 0 & 0 & 0 & -1/(2q^{1/3}) & 0 & 0 \\ 
0 & -1/(2q^{1/3}) & 0 & 1/2 & 0 & 0 & 0 & 0 & 0 \\ 
0 & 0 & 0 & 0 & 0 & 0 & 0 & 0 & 0 \\ 
0 & 0 & 0 & 0 & 0 & 1/2 & 0 & -q^{1/3}/2 & 0 \\ 
0 & 0 & -q^{1/3}/2 & 0 & 0 & 0 & 1/2 & 0 & 0 \\ 
0 & 0 & 0 & 0 & 0 & -1/(2q^{1/3}) & 0 & 1/2 & 0 \\ 
0 & 0 & 0 & 0 & 0 & 0 & 0 & 0 & 0%
\end{array}%
\right) .
\end{align}

\subsection{Fundamental properties of the hierarchy of transfer matrices}

Here, we collect some relevant known properties \cite{KulRS81,KulR83,KunNS94} of the fused transfer
matrices for the representations under consideration.

\begin{proposition}
The transfer matrices:%
\begin{equation}
T_{1}^{\left( K\right) }(\lambda )\equiv \text{tr}_{a}M_{a}^{(K)}(\lambda ),%
\text{ \ \ \ }T_{2}^{\left( K\right) }(\lambda )\equiv \text{tr}%
_{a,b}P_{a,b}^{-}M_{b}^{(K)}(\lambda )M_{a}^{(K)}(\lambda /q),
\end{equation}%
satisfy the following commutation relations:%
\begin{equation}
\left[ T_{1}^{\left( K\right) }(\lambda ),T_{1}^{\left( K\right) }(\mu )%
\right] =\left[ T_{1}^{\left( K\right) }(\lambda ),T_{2}^{\left( K\right)
}(\mu )\right] =\left[ T_{2}^{\left( K\right) }(\lambda ),T_{2}^{\left(
K\right) }(\mu )\right] =0.
\end{equation}%
The quantum determinant:%
\begin{equation}
q\text{-det}M^{(K)}(\lambda )\equiv \text{tr}_{abc}P_{abc}^{-}M_{c}^{(K)}(%
\lambda )M_{b}^{(K)}(\lambda /q)M_{a}^{(K)}(\lambda /q^{2})
\label{Bound-q-detU_1-tri}
\end{equation}%
is a central element of the algebra, i.e.%
\begin{equation}
\lbrack q\text{-det}M^{(K)}(\lambda ),M_{a}^{(K)}(\mu )]=0.
\end{equation}
\end{proposition}

Furthermore, let us define the operators $\mathsf{N}_{i}\in \text{End}(%
\mathcal{H})$ by the following action:%
\begin{equation}
\mathsf{N}_{i}\otimes _{n=1}^{\mathsf{N}}|n,a_{n}\rangle =\otimes _{n=1}^{%
\mathsf{N}}|n,a_{n}\rangle \sum_{n=1}^{\mathsf{N}}\delta _{i,a_{n}} \text{, }
\forall \ i\in \{1,2,3\},
\end{equation}%
in the basis%
\begin{equation}
|n,a_{n}\rangle =\left( 
\begin{array}{c}
\delta _{1,a_{n}} \\ 
\delta _{2,a_{n}} \\ 
\delta _{3,a_{n}}%
\end{array}%
\right) _{n}\in V_{n}\text{ \ }\forall n\in \{1,...,\mathsf{N}\}\text{ and }%
a_{n}\in \{1,2,3\}\text{,}
\end{equation}%
for which we have:%
\begin{equation}
\mathsf{N}_{1}+\mathsf{N}_{2}+\mathsf{N}_{3}=\mathsf{N} \ .
\end{equation}%
These operators generalize to higher rank the symmetry of the $R$-matrix given in the rank one case by the third component of spin. Then it holds:

\begin{proposition}
The transfer matrices $T_{1}^{\left( K\right) }(\lambda )$ and $T_{2}^{\left( K\right) }(\lambda )$ satisfy the following properties:

i) For any $K$-matrix defining a symmetry of the $R$-matrix, $T_{2}^{\left( K\right) }(\lambda )$ has the
following $2 \mathsf{N}$ central zeroes: 
\begin{equation}
T_{2}^{\left( K\right) }(\pm q\xi _{b})=0\text{ \ }\forall b\in \{1,...,%
\mathsf{N}\},
\end{equation}%
the quantum determinant reads:%
\begin{equation}
q\text{-det}M^{(K)}(\lambda )=\text{det}K\text{ }\prod_{b=1}^{\mathsf{N}%
}(\lambda q/\xi _{b}-\xi _{b}/(q\lambda ))(\lambda /(q\xi _{b})-(q\xi
_{b})/\lambda )(\lambda /(q^{2}\xi _{b})-(q^{2}\xi _{b})/\lambda ),
\end{equation}%
and the following fusion identities hold for any $b \in \{1,...,%
\mathsf{N}\}$:%
\begin{eqnarray}
T_{1}^{\left( K\right) }(\xi _{b})T_{1}^{\left( K\right) }(\xi _{b}/q)
&=&T_{2}^{\left( K\right) }(\xi _{b}), \\
T_{1}^{\left( K\right) }(\xi _{b})T_{2}^{\left( K\right) }(\xi _{b}/q) &=&q%
\text{-det}M^{(K)}(\xi _{b}).
\end{eqnarray}
Moreover, in the case of a diagonal twist $K^{(1)}$:

ii) $\lambda ^{\mathsf{N}}T_{1}^{\left( K\right) }(\lambda )$ is a degree 
\textsf{$N$} polynomial in $\lambda ^{2}$ with the following asymptotics:%
\begin{equation}
T_{1}^{(\pm \infty |K)}\equiv \lim_{\log \lambda \rightarrow \pm \infty
}\lambda ^{\mp \mathsf{N}}T_{1}^{\left( K\right) }(\lambda )=\left( \pm
1\right) ^{\mathsf{N}}\frac{\alpha q^{\pm \mathsf{N}_{1}}+\beta q^{\pm 
\mathsf{N}_{2}}+\gamma q^{\pm \mathsf{N}_{3}}}{\prod_{n=1}^{\mathsf{N}}\xi
_{n}^{\pm 1}}.
\end{equation}

iii) $\lambda ^{2\mathsf{N}}T_{2}^{\left( K\right) }(\lambda )$ is a degree $%
2\mathsf{N}$ polynomial in $\lambda ^{2}$ with the following asymptotics:%
\begin{equation}
T_{2}^{(\pm \infty |K)}\equiv \lim_{\log \lambda \rightarrow \pm \infty
}\lambda ^{\mp 2\mathsf{N}}T_{2}^{\left( K\right) }(\lambda )=\frac{\alpha
\beta q^{\pm (\mathsf{N}_{1}+\mathsf{N}_{2})}+\alpha \gamma q^{\pm (\mathsf{N%
}_{1}+\mathsf{N}_{3})}+\beta \gamma q^{\pm (\mathsf{N}_{2}+\mathsf{N}_{3})}}{%
q^{\pm \mathsf{N}}\prod_{n=1}^{\mathsf{N}}\xi _{n}^{\pm 2}}.
\end{equation}

iv) The operators $N_i$ commute with the transfer matrices.\\

In the case of  non-diagonal twists $K^{(2)}$ or $K^{(3)}$:

v) $\lambda ^{(\mathsf{N}-1/3)}T_{1}^{\left( K\right) }(\lambda )$ is a
degree $\mathsf{N}-1$ polynomial in $\lambda ^{2}$.

vi) $\lambda ^{(2\mathsf{N}-2/3)}T_{2}^{\left( K\right) }(\lambda )$ is a
degree $2\mathsf{N}-1$ polynomial in $\lambda ^{2}$.
\end{proposition}

Let us introduce the functions%
\begin{align}
f_{l,\mathbf{h}}^{(a,m)}(\lambda )& =\left( \prod_{b=1}^{\mathsf{N}}\frac{%
\lambda /\xi _{b}^{(-1)}-\xi _{b}^{(-1)}/\lambda }{\xi _{l}^{(h_{l})}/\xi
_{b}^{(-1)}-\xi _{b}^{(-1)}/\xi _{l}^{(h_{l})}}\right) ^{\delta
_{m,2}}\left( \frac{t_{h_{1},...,h_{\mathsf{N}}}\lambda /\xi _{n}+\xi
_{n}/(t_{h_{1},...,h_{\mathsf{N}}}\lambda )}{t_{h_{1},...,h_{\mathsf{N}%
}}+1/t_{h_{1},...,h_{\mathsf{N}}}}\right) ^{\delta _{1,a}}  \notag \\
& \times \prod_{b\neq n,b=1}^{\mathsf{N}}\frac{\lambda /\xi
_{b}^{(h_{b})}-\xi _{b}^{(h_{b})}/\lambda }{\xi _{n}^{(h_{n})}/\xi
_{b}^{(h_{b})}-\xi _{b}^{(h_{b})}/\xi _{n}^{(h_{n})}},\text{\ \ \ \ \ \ }\xi
_{b}^{(h)}\left. =\right. \xi _{b}/q^{h},\text{ }t_{h_{1},...,h_{\mathsf{N}%
}}\left. =\right. q^{-\sum_{a=1}^{\mathsf{N}}h_{a}},
\end{align}%
that  are well defined under the
assumption that $q$ is not a root of unity and $h_{n} \in \{0,...,n-1\}$ for any $n\in \{1,...,\mathsf{N}\}$. We also define the
following functions of the operators $N_i$:
\begin{equation}
T_{1,\mathbf{h}}^{(K^{\left( a\right) },\infty )}(\lambda )=\delta _{1,a}%
\frac{\alpha \cosh \eta \mathsf{N}_{1}+\beta \cosh \eta \mathsf{N}_{2}+\gamma \cosh \eta 
\mathsf{N}_{3}}{\cosh (\eta \sum_{b=1}^{\mathsf{N}}h_{b})}\prod_{b=1}^{%
\mathsf{N}}(\lambda /\xi _{b}^{(h_{b})}-\xi _{b}^{(h_{b})}/\lambda ),
\label{T_1,h-asymp}
\end{equation}%
and%
\begin{align}
T_{2,\mathbf{h}}^{(K^{\left( a\right) },\infty )}(\lambda )& =\delta _{1,a}%
\frac{\alpha \beta \cosh \eta (\mathsf{N}_{1}+\mathsf{N}_{2})+\alpha \gamma
\cosh \eta (\mathsf{N}_{1}+\mathsf{N}_{3})+\beta \gamma \sinh \eta (\mathsf{N%
}_{2}+\mathsf{N}_{3})}{\cosh (\eta \sum_{b=1}^{\mathsf{N}}h_{b})}  \notag \\
& \times \prod_{b=1}^{\mathsf{N}}(\lambda /\xi _{b}^{(h_{b})}-\xi
_{b}^{(h_{b})}/\lambda )(\lambda /\xi _{b}^{(-1)}-\xi _{b}^{(-1)}/\lambda ) \ .
\end{align}%
Then the next corollary holds:

\begin{corollary}
\label{Coro-U-q(gl_3)}The transfer matrix $T_{1}^{(K^{\left( a\right)
})}(\lambda )$ and $T_{2}^{(K^{\left( a\right) })}(\lambda )$ admit the
following interpolation formulae:%
\begin{equation}
T_{1}^{(K^{\left( a\right) })}(\lambda )=T_{1,\mathbf{h}}^{(K^{\left(
a\right) },\infty )}(\lambda )+\sum_{n=1}^{\mathsf{N}}f_{n,\mathbf{h}%
}^{\left( a,1\right) }(\lambda )T_{1}^{(K^{\left( a\right) })}(\xi
_{n}^{(h_{n})}),
\end{equation}%
and%
\begin{equation}
T_{2}^{(K^{\left( a\right) })}(\lambda )=T_{2,\mathbf{h}}^{(K^{\left(
a\right) },\infty )}(\lambda )+\sum_{n=1}^{\mathsf{N}}f_{n,\mathbf{h}%
}^{\left( a,2\right) }(\lambda )T_{2}^{(K^{\left( a\right) })}(\xi
_{n}^{\left( h_{n}\right) }),
\end{equation}%
under the
assumption that $q$ is not a root of unity and $h_{n} \in \{0, 1, 2\}$ for any $n\in \{1,...,\mathsf{N}\}$. Moreover, the following sum rules
are satisfied:%
\begin{align}
& \delta _{1,a}\left( \alpha \sinh \eta (\mathsf{N}_{1}-\sum_{b=1}^{\mathsf{N%
}}h_{b})+\beta \sinh \eta (\mathsf{N}_{2}-\sum_{b=1}^{\mathsf{N}%
}h_{b})+\gamma \sinh \eta (\mathsf{N}_{3}-\sum_{b=1}^{\mathsf{N}%
}h_{b})\right)  \notag \\
& =\sum_{n=1}^{\mathsf{N}}\frac{T_{1}^{(K^{(a)})}(\xi _{n}^{\left(
h_{n}\right) })}{2\prod_{b\neq n,b=1}^{\mathsf{N}}\xi _{n}^{\left(
h_{n}\right) }/\xi _{b}^{\left( h_{b}\right) }-\xi _{b}^{\left( h_{b}\right)
}/\xi _{n}^{\left( h_{n}\right) }},
\end{align}%
and%
\begin{align}
&\delta _{1,a}\left( \alpha \beta \sinh \eta (\mathsf{N}_{1}+\mathsf{N}%
_{2}-\sum_{b=1}^{\mathsf{N}}h_{b})+\alpha \gamma \sinh \eta (\mathsf{N}_{1}+%
\mathsf{N}_{3}-\sum_{b=1}^{\mathsf{N}}h_{b})+\beta \gamma \sinh \eta (%
\mathsf{N}_{2}+\mathsf{N}_{3}-\sum_{b=1}^{\mathsf{N}}h_{b})\right)  \notag \\
& =\sum_{n=1}^{\mathsf{N}}\frac{T_{2}^{(K^{(a)})}(\xi _{n}^{\left(
h_{n}\right) })}{2\prod_{b\neq n,b=1}^{\mathsf{N}}(\xi _{n}^{\left(
h_{n}\right) }/\xi _{b}^{\left( h_{b}\right) }-\xi _{b}^{\left( h_{b}\right)
}/\xi _{n}^{\left( h_{n}\right) })\prod_{b=1}^{\mathsf{N}}(\xi _{n}^{\left(
h_{n}\right) }/\xi _{b}^{\left( -1\right) }-\xi _{b}^{\left( -1\right) }/\xi
_{n}^{\left( h_{n}\right) })}.
\end{align}%
$T_{1}^{(K^{\left( a\right) })}(\lambda )$ then completely characterizes $%
T_{2}^{(K^{\left( a\right) })}(\lambda )$ in terms of the fusion equations
by:%
\begin{equation}
T_{2}^{(K^{\left( a\right) })}(\lambda )=T_{2,\mathbf{h=0}}^{(K^{\left(
a\right) },\infty )}(\lambda )+\sum_{n=1}^{\mathsf{N}}f_{n,\mathbf{h=0}%
}^{\left( a,2\right) }(\lambda )T_{1}^{(K^{\left( a\right) })}(\xi
_{n}/q)T_{1}^{(K^{\left( a\right) })}(\xi _{n}).
\end{equation}
\end{corollary}

\begin{proof}
We have to use just the known central zeroes and asymptotic behavior to prove the
above interpolation formula once $T_{2}^{(K^{\left( a\right) })}(\xi _{n})$
is given by the fusion equations. The sum rules just follow from the fact that in this trigonometric case we know the asymptotic behavior of the transfer matrices in two points (at $\lambda$ going to zero and to infinity) while we still reconstruct these degree $\mathsf{N}$ polynomials in $\mathsf{N}$  points which leads to the sum rule.
\end{proof}

\subsection{SoV covector basis generated by transfer matrix action}

The general Proposition 2.6 of our article \cite{MaiN18} for the
construction of the SoV covector basis applies in particular to the
fundamental representations of the trigonometric
Yang-Baxter algebra.

The twist  $K^{\left( a\right) }$ are diagonalizable and with simple spectrum  $3\times 3$
matrices, as soon as $\alpha ,$ $\beta $ and $\gamma $ are all different in the
case $a=1$ and $\alpha \beta \gamma \neq 0$ in the case $a=2$ and $a=3$. Let
us denote by $K_{J}$ the diagonal form of the matrix $K$ and $W_K$ the
invertible matrix defining the change of basis:%
\begin{equation}
K=W_{K}K_{J}W_{K}^{-1}\text{ \ with \ }K_{J}=\left( 
\begin{array}{ccc}
\mathsf{k}_{1} & 0 & 0 \\ 
0 & \mathsf{k}_{2} & 0 \\ 
0 & 0 & \mathsf{k}_{3}%
\end{array}%
\right) .
\end{equation}%
Then the following theorem holds:

\begin{theorem}
For any diagonalizable $3\times 3$ twist matrix $K$ having simple spectrum, the following set:%
\begin{equation}
\langle h_{1},...,h_{\mathsf{N}}|\equiv \langle S|\prod_{n=1}^{\mathsf{N}%
}(T_{1}^{(K)}(\xi _{n}))^{h_{n}}\text{ \ for any }\{h_{1},...,h_{\mathsf{N}%
}\}\in \{0,1,2\}^{\otimes \mathsf{N}},
\end{equation}%
forms a covector basis of $\mathcal{H}$, for almost any choice of $\langle
S| $ and of the inhomogeneities under the condition 
\begin{equation}
\xi _{a}\neq \xi _{b}q^h \,\,\,\,\,\forall a\neq b\in \{1,...,\mathsf{N}\} 
\text{ and } h\in\{-2,-1,0,1,2\}.  \label{Inhomog-cond}
\end{equation}%
Moreover, a proper choice of the state $\langle S|$ has the following tensor
product form:%
\begin{equation}
\langle S|=\bigotimes_{a=1}^{\mathsf{N}}(x,y,z)_{a}\Gamma _{W}^{-1},\text{ \
\ }\Gamma _{W}=\bigotimes_{a=1}^{\mathsf{N}}W_{K,a}
\end{equation}%
under the condition $x\,y\,z\neq 0$.
\end{theorem}

\begin{proof}
As a corollary of the general Proposition 2.6 of \cite{MaiN18}, this set of
covectors is a covector basis of $\mathcal{H}$ as soon as we can show that
the covectors:%
\begin{equation}
(x,y,z)_{a}W^{-1},(x,y,z)_{a}W^{-1}K,(x,y,z)_{a}W^{-1}K^{2},
\end{equation}%
or equivalently:%
\begin{equation}
(x,y,z)_{a},(x,y,z)_{a}K_{J},(x,y,z)_{a}K_{J}^{2},
\end{equation}%
form a basis in $V_{a}\cong \mathbb{C}^{3}$, that is that the following
determinant is non-zero:%
\begin{equation}
\text{det}||\left( (x,y,z)K_{J}^{i-1}e_{j}(a)\right) _{i,j\in
\{1,2,3\}}||=-xyzV(\mathsf{k}_{0},\mathsf{k}_{1},\mathsf{k}_{2}),
\end{equation}%
which leads to the given requirements on the components $x,y,z\in \mathbb{C}$
of the three dimensional covector.
\end{proof}

\section{Transfer matrix spectrum in our SoV
approach: the $\mathcal{U}_{q}(\widehat{gl_{3}})$ case}

\subsection{Discrete spectrum characterization by SoV}

For any given twist matrix $K$ diagonalizable and with simple spectrum, the following characterization of the transfer matrix spectrum holds:

\begin{theorem}
\label{Discrete-Ch-3_q}The spectrum of $T_{1}^{(K)}(\lambda )$ is
characterized by:%
\begin{equation}
\Sigma _{T^{(K)}}=\bigcup_{0\leq
l+m\leq \mathsf{N}}\Sigma _{T^{(K)}}^{\left( l,m\right) },
\end{equation}%
where $l, m$ are positive integers and %
\begin{equation}
\Sigma _{T^{(K)}}^{\left( l,m\right) }=\left\{ t_{1}(\lambda
)=t_{1}(l,m,\alpha,\beta,\gamma)\prod_{b=1}^{\mathsf{N}}(\lambda /\xi
_{b}-\xi _{b}/\lambda )+\sum_{a=1}^{\mathsf{N}}f_{n,\mathbf{h=0}}^{\left(
a,1\right) }(\lambda )x_{a} \right\},  \label{SET-T}
\end{equation}%
for any set $\{x_1, \dots, x_{\mathsf{N}} \}$ belonging to  $S_{T^{(K)}}^{\left( l,m\right)}$ where we have defined: 
\begin{equation}
t_{1}(l,m,\alpha,\beta,\gamma)\equiv\delta _{1,a}\left( \alpha \cosh \eta l+\beta
\cosh \eta m+\gamma \cosh \eta (l+m-\mathsf{N})\right)
\end{equation}
under the assumption that the $3\times 3$ twist matrix $K$ is simple and
diagonalizable and the inhomogeneities satisfy $(\ref{Inhomog-cond})$. Here, 
$S_{T^{(K)}}^{\left( l,m\right) }$ stand for the set of solutions $\{x_1, \dots,x_{\mathsf{N}}\}$ to the
following system of $\mathsf{N}$ cubic equations:%
\begin{align}
& x_{n}[\delta _{1,a}\left( \alpha \beta \cosh \eta (l+m)+\alpha \gamma
\cosh \eta (\mathsf{N}-m)+\beta \gamma \sinh \eta (\mathsf{N}-l)\right)
\prod_{b=1}^{\mathsf{N}}(\xi _{n}^{(1)}/\xi _{b}-\xi _{b}/\xi _{n}^{(1)}) 
\notag \\
& \times (\xi _{n}^{(1)}/\xi _{b}^{(-1)}-\xi _{b}^{(-1)}/\xi
_{n}^{(1)})+\sum_{m=1}^{\mathsf{N}}f_{m,\mathbf{h=0}}^{\left( a,2\right)
}(\xi _{m}^{(1)})t_{1}(\xi _{m}^{(1)})x_{m}]\left. =\right. q\text{-det}%
M^{(K)}(\xi _{a}),
\end{align}%
in $\mathsf{N}$ unknown $\{x_{1},...,x_{\mathsf{N}}\}$. Moreover, $%
T_{1}^{(K)}(\lambda )$ is diagonalizable with simple spectrum and for any $%
t_{1}(\lambda )\in \Sigma _{T^{(K)}}$ the associated unique eigenvector $%
|t\rangle $ has the following wave function in the covector SoV basis:%
\begin{equation}
\langle h_{1},...,h_{\mathsf{N}}|t\rangle =\prod_{n=1}^{\mathsf{N}%
}t_{1}^{h_{n}}(\xi _{n}),  \label{SoV-Ch-T-eigenV}
\end{equation}%
where the overall normalization has been fixed by imposing $\langle
S|t\rangle =1$.
\end{theorem}

\begin{proof}
The transfer matrix fusion equations:%
\begin{equation}
T_{1}^{\left( K\right) }(\xi _{b})T_{2}^{\left( K\right) }(\xi _{b}/q)=q%
\text{-det}M^{(K)}(\xi _{b}),\text{ }\forall b\in \{1,...,\mathsf{N}\},
\label{scalar-fusion-1}
\end{equation}%
when rewritten for the eigenvalues take exactly the form of the given system
of $\mathsf{N}$ cubic equations in $\mathsf{N}$ unknowns $\{x_{1},...,x_{%
\mathsf{N}}\}$, once we use the known analyticity and central zeroes.
Consequently, this system has to be satisfied and the given characterization
of any eigenvector $|t\rangle $ is implied.

The reverse statement has to be shown now. In particular, we have to prove
that given a polynomial of the above form satisfying this system then it is
an eigenvalue of the transfer matrix $T_{1}^{(K)}(\lambda )$. This is
done by proving that the vector $|t\rangle $ defined in $(\ref%
{SoV-Ch-T-eigenV})$ is a transfer matrix eigenvector using our SoV basis:%
\begin{equation}
\langle h_{1},...,h_{\mathsf{N}}|T_{1}^{(K)}(\lambda )|t\rangle
=t_{1}(\lambda )\langle h_{1},...,h_{\mathsf{N}}|t\rangle ,\text{ }\forall
\{h_{1},...,h_{\mathsf{N}}\}\in \{0,1,2\}^{\otimes \mathsf{N}}.
\label{Eigen-cond-sl3}
\end{equation}%
Let us point out that as a consequence of the Corollary \ref{Coro-U-q(gl_3)}%
, in order to prove this identity it is enough to prove it in a generic $%
\mathsf{N}$-upla of points $\xi _{a}^{(k_{a})}$ for any $a\in \{1,...,%
\mathsf{N}\}$. Indeed, let us assume that it holds:%
\begin{equation}
\langle h_{1},...,h_{\mathsf{N}}|T_{1}^{(K)}(\xi _{a}^{(k_{a})})|t\rangle
=t_{1}(\xi _{a}^{(k_{a})})\langle h_{1},...,h_{\mathsf{N}}|t\rangle ,\text{ }%
\forall \{h_{1},...,h_{\mathsf{N}}\}\in \{0,1,2\}^{\otimes \mathsf{N}},
\end{equation}%
then we have that:%
\begin{align}
& \delta _{1,a}\langle h_{1},...,h_{\mathsf{N}}|\left( \alpha \sinh \eta (%
\mathsf{N}_{1}-\sum_{b=1}^{\mathsf{N}}k_{b})+\beta \sinh \eta (\mathsf{N}%
_{2}-\sum_{b=1}^{\mathsf{N}}k_{b})+\gamma \sinh \eta (\mathsf{N}%
_{3}-\sum_{b=1}^{\mathsf{N}}k_{b})\right) |t\rangle  \notag \\
& =\langle h_{1},...,h_{\mathsf{N}}|\left( \sum_{n=1}^{\mathsf{N}}\frac{%
T_{1}^{(K^{(a)})}(\xi _{n}^{\left( k_{n}\right) })}{2\prod_{b\neq n,b=1}^{%
\mathsf{N}}\xi _{n}^{\left( k_{n}\right) }/\xi _{b}^{\left( k_{b}\right)
}-\xi _{b}^{\left( k_{b}\right) }/\xi _{n}^{\left( k_{n}\right) }}\right)
|t\rangle \\
& =\left( \sum_{n=1}^{\mathsf{N}}\frac{t_{1}(\xi _{n}^{\left( k_{n}\right) })%
}{2\prod_{b\neq n,b=1}^{\mathsf{N}}\xi _{n}^{\left( k_{n}\right) }/\xi
_{b}^{\left( k_{b}\right) }-\xi _{b}^{\left( k_{b}\right) }/\xi _{n}^{\left(
k_{n}\right) }}\right) \langle h_{1},...,h_{\mathsf{N}}|t\rangle \\
& =\delta _{1,a}( \alpha \sinh \eta (l-\sum_{b=1}^{\mathsf{N}}k_{b})+\beta
\sinh \eta (m-\sum_{b=1}^{\mathsf{N}}k_{b})+\gamma \sinh \eta (\mathsf{N}%
-(l+m)-\sum_{b=1}^{\mathsf{N}}k_{b})) \langle h_{1},...,h_{\mathsf{N}%
}|t\rangle ,
\end{align}%
which implies that for $a=1$,  $|t\rangle $ is an eigenvector of the charges $\mathsf{N}%
_{1}$, $\mathsf{N}_{2}$ and $\mathsf{N}_{3}$ with eigenvalues, respectively, 
$l,$ $m$ and $\mathsf{N}-(l+m)$ which in turn fix the asymptotics of the
transfer matrices in that case. That is for $a=1$, any $t_{1}(\lambda )$ in $\left( \ref{SET-T}%
\right) $, $t_{1}(\lambda )$ and $|t\rangle $ can be transfer matrix
eigenvalue and eigenvector only associated to the common eigenspace of $%
\mathsf{N}_{1}$ and $\mathsf{N}_{2}$ corresponding to the nonnegative integer 
eigenvalues $l$ and $m$, respectively. Notice that if $a \neq 1$ the asymptotic term is zero and the $\mathsf{N}_i$ are no longer symmetries of the transfer matrices and so we don't need to distinguish those values in the discussion. 

Let $h_{a}=0,1$ and $h_{b}\in \{0,1,2\}$ for any $b\in \{1,...,\mathsf{N}%
\}\backslash a$, then we have the following identities:%
\begin{align}
\langle h_{1},...,h_{\mathsf{N}}|T_{1}^{(K^{\left( a\right) })}(\xi
_{a})|t\rangle & =\langle h_{1},...,h_{a}+1,...,h_{\mathsf{N}}|t\rangle 
\notag \\
& =t_{1}(\xi _{a})\langle h_{1},...,h_{a},...,h_{\mathsf{N}}|t\rangle ,
\label{Id-step1}
\end{align}%
as a direct consequence of the definition of the covector SoV basis and of
the state $|t\rangle $. So that we are left with the proof of the statement
in the case $h_{a}=2$. In this case we want to prove that it holds:%
\begin{equation}
\langle h_{1},...,h_{\mathsf{N}}|T_{1}^{(K^{\left( a\right) })}(\xi
_{a}/q)|t\rangle =t_{1}(\xi _{a}/q)\langle h_{1},...,h_{a},...,h_{\mathsf{N}%
}|t\rangle ,  \label{Id-step2}
\end{equation}%
the proof is done by induction on the number of zeros contained in $%
\{h_{1},...,h_{\mathsf{N}}\}\in \{0,1,2\}^{\otimes \mathsf{N}}$. It is
developed just following the same steps we have developed in the case of the
fundamental representation of the $Y(gl_{3})$\ rational Yang-Baxter algebra
in our paper \cite{MaiN18}. In fact we have only to take into account
the different functional form of the transfer matrix, i.e. they are Laurent
polynomials and not simple polynomials, and the fact that the asymptotic behavior of
the transfer matrices are not central for $a=1$ but take fixed values in any common eigenspace of $%
\mathsf{N}_{1}$ and $\mathsf{N}_{2}$ that is stable by the action of the transfer matrices since for $a=1$ they commute with each $\mathsf{N}_{i}$. For completeness we dedicate Appendix A to make explicit these steps of the proof.
\end{proof}

\subsection{Spectrum characterization by quantum spectral curve}

The discrete characterization of the transfer matrix spectrum derived in the
previous section in our SoV basis allows us to introduce the following
quantum spectral curve functional reformulation. Let us first introduce the
functions:%
\begin{eqnarray}
\delta _{3}(\lambda ) &=&\delta _{1}(\lambda )\delta _{1}(\lambda /q)\delta
_{1}(\lambda /q^{2}), \\
\delta _{2}(\lambda ) &=&\delta _{1}(\lambda )\delta _{1}(\lambda /q), \\
\delta _{1}(\lambda ) &=&\delta _{0}\prod_{a=1}^{\mathsf{N}}(\lambda q/\xi
_{a}-\xi _{a}/(\lambda q)),
\end{eqnarray}

\begin{theorem}
\label{Functio-Ch-3_q}Here we consider the case of a twist $K$ which is diagonal and with 
simple spectrum, then the entire functions $t_{1}(\lambda )$ is a $T_{1}^{\left(
K\right) }(\lambda )$ transfer matrix eigenvalue belonging to $\Sigma
_{T^{(K)}}^{\left( \nu_{1},\nu_{2}\right) }$, with $\nu_{1},$ $\nu_{2}$ two nonnegative
integers satisfying:%
\begin{equation}
\nu_{1}+\nu_{2}\leq\mathsf{N},
\end{equation}%
iff there exists an unique Laurent polynomial of the form:%
\begin{equation}
\varphi _{t}(\lambda )=\prod_{a=1}^{\mathsf{M}}(\lambda /\lambda
_{a}-\lambda _{a}/\lambda )\text{\ \ with }\mathsf{M}\leq \mathsf{N}\text{
and }\lambda _{a}\neq \xi _{n}\text{ }\forall (a,n)\in \{1,...,\mathsf{M}%
\}\times \{1,...,\mathsf{N}\},  \label{Phi-form}
\end{equation}%
solution of the following quantum spectral curve functional equation:%
\begin{equation}
\delta _{3}(\lambda )\varphi _{t}(\lambda /q^{3})-\delta _{2}(\lambda
)t_{1}(\lambda /q^{2})\varphi _{t}(\lambda /q^{2})+\delta _{1}(\lambda
)t_{2}(\lambda /q)\varphi _{t}(\lambda /q)-q\text{-det}M_{a}^{(K)}(\lambda
)\varphi _{t}(\lambda )\left. =\right. 0
\end{equation}%
where we have defined%
\begin{align}
t_{2}(\lambda )& =\left( \alpha \beta \cosh \eta (\nu_{1}+\nu_{2})+\alpha \gamma
\cosh \eta (\nu_{1}+\nu_{3})+\beta \gamma \sinh \eta (\nu_{2}+\nu_{3})\right)  \notag
\\
& \times \prod_{b=1}^{\mathsf{N}}(\lambda /\xi _{b}-\xi _{b}/\lambda
)(\lambda /\xi _{b}^{(-1)}-\xi _{b}^{(-1)}/\lambda )+\sum_{n=1}^{\mathsf{N}%
}f_{n,\mathbf{h=0}}^{\left( a\right) }(\xi _{n}^{(1)})t_{1}(\xi
_{n}^{(1)})t_{1}(\xi _{n}),
\end{align}%
with $\nu_3 =\mathsf{N} -\nu_1 -\nu_2$ and fixed\footnote{%
That is we have to fix $\delta _{0}$ to be one of the three distinct
eigenvalues of the matrix\thinspace $K$ and then the degree of the Laurent
polynomial $\varphi _{t}(\lambda )$ is fixed by:%
\begin{equation}
\mathsf{M}=\mathsf{N}-\nu_{1},\text{ for }\delta _{0}=\alpha ,\text{ }\mathsf{M%
}=\mathsf{N}-\nu_{2},\text{ for }\delta _{0}=\beta ,\text{ }\mathsf{M}=\mathsf{%
N}-\nu_{3},\text{ for }\delta _{0}=\gamma \text{, with }\nu_{3}=\mathsf{N}-(\nu_{1}+\nu_{2}).\notag
\end{equation}%
}:%
\begin{equation}
\delta _{0}=\mathsf{k}_{i}\text{ for one fixed }i\in \{1,2,3\},
\end{equation}%
with%
\begin{equation}
\mathsf{M}=\mathsf{M}-\nu_{i}.
\end{equation}%
Moreover, up to an overall normalization the common transfer matrix
eigenvector $|t\rangle $ admits the following separate representation:%
\begin{equation}
\langle h_{1},...,h_{\mathsf{N}}|t\rangle =\prod_{a=1}^{\mathsf{N}}\delta
_{1}^{h_{a}}(\xi _{a})\varphi _{t}^{h_{a}}(\xi _{a}/q)\varphi
_{t}^{2-h_{a}}(\xi _{a}).
\end{equation}
\end{theorem}

\begin{proof}
This is a special case of the proof that will be given in the general $U_{q}(%
\widehat{gl_{n}})$ case.
\end{proof}

Let us comment that the Theorem \ref{Discrete-Ch-3_q} applies for any
integrable boundary conditions while the previous Theorem \ref%
{Functio-Ch-3_q} applies only to the case of diagonal twists. The
non-diagonal case is not presented explicitly but we can similarly derive a
functional equation of third order type which however is of inhomogeneous
type, if we ask that the $\varphi _{t}(\lambda )$ has the same Laurent
polynomial form as indicated in $\left( \ref{Phi-form}\right) $. While at
this stage this is a simple exercise we think that the main interesting
question to investigate about the non-diagonal twists is if, with a different
(periodicity) definition of the $\varphi _{t}(\lambda )$ function, we can
reestablish an homogeneous equation as it has been proven for the $U_{q}(%
\widehat{gl_{2}})$ case in \cite{NicT15}.

\subsection{ABA rewriting of transfer matrix eigenvectors}

An equivalent rewriting of algebraic Bethe ansatz type for the transfer
matrix eigenvectors can be derived on the basis of their SoV representation.
Let us first define one common eigenvector of the transfer matrix $%
T_{1}^{\left( K\right) }(\lambda )$ and $T_{2}^{\left( K\right) }(\lambda )$:

\begin{lemma}
Let $K$ be a diagonal $3\times 3$ matrix having three distinct eigenvalues $\mathsf{k}_i$:%
\begin{equation}
K=\left( 
\begin{array}{ccc}
\mathsf{k}_{1} & 0 & 0 \\ 
0 & \mathsf{k}_{2} & 0 \\ 
0 & 0 & \mathsf{k}_{3}%
\end{array}%
\right) ,
\end{equation}%
then:%
\begin{equation}
|t_{0}\rangle =\bigotimes_{a=1}^{\mathsf{N}}\left( 
\begin{array}{c}
1 \\ 
0 \\ 
0%
\end{array}%
\right) _{a},
\end{equation}%
is a common eigenvector of the transfer matrices $T_{1}^{\left( K\right)
}(\lambda )$ and $T_{2}^{\left( K\right) }(\lambda )$:%
\begin{align}
T_{1}^{\left( K\right) }(\lambda )|t_{0}\rangle & =|t_{0}\rangle
t_{1,0}(\lambda )\text{ \ with }t_{1,0}(\lambda )=\mathsf{k}_{1}\prod_{a=1}^{%
\mathsf{N}}(\lambda q/\xi _{a}-\xi _{a}/(\lambda q))+(\mathsf{k}_{2}+\mathsf{%
k}_{3})\prod_{a=1}^{\mathsf{N}}(\lambda /\xi _{a}-\xi _{a}/\lambda ), \\
T_{2}^{\left( K\right) }(\lambda )|t_{0}\rangle & =|t_{0}\rangle
t_{2,0}(\lambda )\text{ \ with }t_{2,0}(\lambda )=\left\{ 
\begin{array}{l}
\prod_{a=1}^{\mathsf{N}}(\lambda /(q\xi _{a})-(q\xi _{a})/\lambda )(\mathsf{k%
}_{3}\mathsf{k}_{2}\prod_{a=1}^{\mathsf{N}}(\lambda /\xi _{a}-\xi
_{a}/\lambda ) \\ 
+(\mathsf{k}_{2}\mathsf{k}_{1}+\mathsf{k}_{3}\mathsf{k}_{1})\prod_{a=1}^{%
\mathsf{N}}(\lambda q/\xi _{a}-\xi _{a}/(\lambda q)),%
\end{array}%
\right.
\end{align}%
and the quantum spectral curve%
\begin{equation}
\delta _{3}(\lambda )-\delta _{2}(\lambda )t_{1}(\lambda /q^{2})+\delta
_{1}(\lambda )t_{2}(\lambda /q)-q\text{-det}M_{a}^{(K)}(\lambda )\left.
=\right. 0
\end{equation}
with constant $\varphi _{t}(\lambda )$ is satisfied by the couple of
eigenvalues $t_{1,0}(\lambda )$ and $t_{2,0}(\lambda )$ for $\delta _{0}=%
\mathsf{k}_{1}.$
\end{lemma}

\begin{proof}
This is a standard result which follows by proving that it holds:%
\begin{equation}
A_{i}^{(I)}(\lambda )|t_{0}\rangle =|t_{0}\rangle \prod_{a=1}^{\mathsf{N}%
}(\lambda q^{\delta _{i,1}}/\xi _{a}-\xi _{a}/(\lambda q^{\delta _{i,1}})),%
\text{ \ }C_{i}^{(I)}(\lambda )|t_{0}\rangle =0,\text{ \ }i\in \{1,2,3\},
\end{equation}%
where the upper index $I$ stands for the identity twist matrix, from which it
is simple to verify by direct computation that the $t_{1,0}(\lambda )$ and $%
t_{2,0}(\lambda )$ satisfies the fusion equations $(\ref{scalar-fusion-1})$
and that it holds:%
\begin{eqnarray}
t_{1,0} &\equiv &\lim_{\lambda \rightarrow \pm \infty }\lambda ^{\mp \mathsf{%
N}}t_{1,0}(\lambda )=\left( -1\right) ^{\frac{1\pm 1}{2}\mathsf{N}}\left( 
\mathsf{k}_{1}q^{\pm \mathsf{N}}+\mathsf{k}_{2}+\mathsf{k}_{3}\right) \left(
\prod_{a=1}^{\mathsf{N}}\xi _{a}\right) ^{\mp 1}, \\
t_{2,0} &\equiv &\lim_{\lambda \rightarrow \pm \infty }\lambda ^{\mp 2%
\mathsf{N}}t_{2,0}(\lambda )=\left( \prod_{a=1}^{\mathsf{N}}\xi _{a}\right)
^{\mp 2}q^{\mp \mathsf{N}}(\mathsf{k}_{3}\mathsf{k}_{2}+(\mathsf{k}_{2}%
\mathsf{k}_{1}+\mathsf{k}_{3}\mathsf{k}_{1})q^{\pm \mathsf{N}}),
\end{eqnarray}%
so that $t_{1,0}(\lambda )$ satisfies the SoV characterization of the
eigenvalues of $T_{1}^{\left( K\right) }(\lambda )$. Observing now that it
holds:%
\begin{equation}
t_{1,0}(\xi _{a})=\delta _{1}(\xi _{a})\text{ \ \ for any }a\in \{1,...,%
\mathsf{N}\}
\end{equation}%
it follows that the associated $\varphi _{t}(\lambda )$ satisfies the
equations:%
\begin{equation}
\varphi _{t}(\xi _{a})=\varphi _{t}(\xi _{a}/q)\text{ \ \ for any }a\in
\{1,...,\mathsf{N}\}
\end{equation}%
and so $\varphi _{t}(\lambda )$ is a constant.
\end{proof}

In our SoV basis we can now define the operator $\mathbb{B}^{(K)}\left( \lambda \right) $ as
the one parameter family of commuting operators through the following
characterization:%
\begin{equation}
\langle h_{1},...,h_{\mathsf{N}}|\mathbb{B}^{(K)}\left( \lambda \right) =%
\text{ }b_{h_{1},...,h_{\mathsf{N}}}(\lambda )\langle h_{1},...,h_{\mathsf{N}%
}|,
\end{equation}%
where we have defined%
\begin{equation}
b_{h_{1},...,h_{\mathsf{N}}}(\lambda )=\prod_{a=1}^{\mathsf{N}}(\lambda /\xi
_{a}-\xi _{a}/\lambda )^{2-h_{a}}(\lambda q/\xi _{a}-\xi _{a}/(q\lambda
))^{h_{a}} \ .
\end{equation}%
Then the following corollary holds:

\begin{lemma}
The following algebraic Bethe ansatz type representation%
\begin{equation}
|t\rangle =\prod_{a=1}^{\mathsf{M}}\mathbb{B}^{(K)}(\lambda
_{a})|t_{0}\rangle \text{\ \ with }\mathsf{M}\leq \mathsf{N}\text{ and }%
\lambda _{a}\neq \xi _{n}\text{ }\forall (a,n)\in \{1,...,\mathsf{M}\}\times
\{1,...,\mathsf{N}\},
\end{equation}%
holds for the unique (up to trivial scalar multiplication) eigenvector $|t\rangle $ associated to any given $%
t_{1}(\lambda )\in \Sigma _{T^{(K)}}\equiv \bigcup_{\forall \nu_{i}\geq 0\text{
: }\nu_{1}+\nu_{2}\leq \mathsf{N}}\Sigma _{T^{(K)}}^{\left( \nu_{1},\nu_{2}\right) }$%
. Here the $\lambda _{a}$ are the roots of the unique Laurent polynomial $%
\varphi _{t}(\lambda )$ associated to $t_{1}(\lambda )\in \Sigma _{T^{(K)}}$.
\end{lemma}

\begin{proof}
The proof is standard, the following chain of identities holds%
\begin{align}
\langle h_{1},...,h_{\mathsf{N}}|\prod_{a=1}^{\mathsf{M}}\mathbb{B}%
^{(K)}(\lambda _{a})|t_{0}\rangle & =\prod_{j=1}^{\mathsf{M}}b_{h_{1},...,h_{%
\mathsf{N}}}(\lambda _{j})\text{ }\langle h_{1},...,h_{\mathsf{N}%
}|t_{0}\rangle  \notag \\
& =\prod_{j=1}^{\mathsf{M}}\prod_{a=1}^{\mathsf{N}}(\lambda _{j}/\xi
_{a}-\xi _{a}/\lambda _{j})^{2-h_{a}}(\lambda _{j}q/\xi _{a}-\xi
_{a}/(q\lambda _{j}))^{h_{a}}\prod_{a=1}^{\mathsf{N}}\delta _{1}^{h_{a}}(\xi
_{a})  \notag \\
& =\prod_{a=1}^{\mathsf{N}}\delta _{1}^{h_{a}}(\xi _{a})\varphi _{t}(\xi
_{a})^{2-h_{a}}\varphi _{t}(\xi _{a}/q)^{h_{a}},
\end{align}%
which coincides with the SoV characterization of the transfer matrix
eigenvector once we recall that:%
\begin{equation}
\langle h_{1},...,h_{\mathsf{N}}|t_{0}\rangle =\prod_{a=1}^{\mathsf{N}%
}\delta _{1}^{h_{a}}(\xi _{a}).
\end{equation}
\end{proof}

\section{Transfer matrices for fundamental evaluation representations of $\mathcal{U}_{q}(\widehat{gl_{n}})$
}

Let us consider now the general higher rank $n-1$, with $n\geq 3$ case. In
particular, here we take the following $R$-matrix:%
\begin{align}
R_{a,b}(\lambda )& =\left( \frac{\lambda }{q}-\frac{q}{\lambda }\right)
\sum_{k=1}^{n}E_{kk}^{(a)}\otimes E_{kk}^{(b)}+\left( \lambda -\frac{1}{%
\lambda }\right) \sum_{p=1}^{n}\sum_{k=1,k\neq p}^{n}E_{kk}^{(a)}\otimes
E_{pp}^{(b)}  \notag \\
& +\left( q-\frac{1}{q}\right) \sum_{1\leq k<p\leq n}\left( \lambda
^{(n-2(p-k))/n}E_{kp}^{(a)}\otimes E_{pk}^{(b)}+\lambda
^{-(n-2(p-k))/n}E_{pk}^{(a)}\otimes E_{kp}^{(b)}\right) \in \text{End}%
(V_{a}\otimes V_{b})
\end{align}%
which is the \textit{trigonometric} principal gradation solution\footnote{%
These $R$-matrices are associated to general values of $q$, the root of unity
case has been also analyzed see for example \cite{Au-YP06} for a review. However in all the present article we assume that $q$ is not a root of unity.} 
\cite{BabdeVV81, Jim86} of the Yang-Baxter equation:%
\begin{equation}
R_{12}(\lambda /\mu )R_{13}(\lambda )R_{23}(\mu )=R_{23}(\mu )R_{13}(\lambda
)R_{12}(\lambda /\mu )\in \text{End}(V_{1}\otimes V_{2}\otimes V_{3}),
\end{equation}%
where $V_{i}\simeq \mathbb{C}^{n}$ for any $i\in \{1,2,3\}$, and it is
associated to the fundamental evaluation representations of $\mathcal{U}_{q}(\widehat{gl_{n}})$ \cite%
{KulR83a,KulR82,Jim85,Dri87,ChaP94}. Above, we have used the standard
notation for the elementary matrices $E_{lm}\in $End$(V\simeq \mathbb{C}^{n})$, $%
(l,m)\in \{1,...,n\}\times \{1,...,n\}$:%
\begin{equation}
\left( E_{lm}\right) _{\alpha \beta }=\delta _{\alpha l}\delta _{m\beta }%
\text{ \ }\forall (\alpha ,\beta )\in \{1,...,n\}\times \{1,...,n\}.
\end{equation}%
In this paper, we analyze the fundamental representations of these rank $n-1$
trigonometric Yang-Baxter algebras, associated to the following monodromy
matrices:%
\begin{equation}
M_{a}^{(K)}(\lambda )\equiv K_{a}R_{a,\mathsf{N}}
(\lambda /\xi _{\mathsf{N}})\cdots R_{a,1}(\lambda
/\xi _{1})\in \text{End}(V_{a}\otimes \mathcal{H}),
\end{equation}%
where $\mathcal{H}=\bigotimes_{n=1}^{\mathsf{N}}V_{n}$ and $K\in $End$(V)$
is a symmetry (twist matrix):%
\begin{equation}
R_{12}(\lambda )K_{1}K_{2}=K_{2}K_{1}R_{12}(\lambda )\in \text{End}%
(V_{1}\otimes V_{2}).
\end{equation}%
Then the one parameter family of operators%
\begin{equation}
T_{1}^{\left( K\right) }(\lambda )\equiv \text{tr}_{a}M_{a}^{\left( K\right)
}(\lambda )\in \text{End}\mathcal{H},
\end{equation}%
are the associated commuting transfer matrices. Here, we focus our attention
on the case of diagonal quasi-periodic boundary conditions which are
associated to generic diagonal $n\times n$ twist matrices having pairwise distinct eigenvalues $\mathsf{k}_{i}$:%
\begin{equation}
K=\left( 
\begin{array}{cccc}
\mathsf{k}_{1} & 0 & \cdots & 0 \\ 
0 & \mathsf{k}_{2} & \ddots & \vdots \\ 
\vdots & \ddots & \ddots & 0 \\ 
0 & \cdots & 0 & \mathsf{k}_{n}%
\end{array}%
\right) .
\end{equation}%
For this class of representations we will prove the complete spectrum
characterization of the transfer matrix in terms of a specific class of
polynomial solutions to the quantum spectral curve equation, an homogeneous
functional equation to the finite difference of order $n$. Let us comment
that, as for the case $n=3$, the symmetry of the principal gradation $R$%
-matrix extends also to non-diagonal twist matrices\footnote{%
In the special case associated to the so-called anti-periodic boundary
conditions a first eigenvalue analysis has been developed in \cite%
{HaoCLYSW16}.} and that our construction of the SoV basis applies for
these cases, as well as the derivation of the complete SoV characterization
of the transfer matrix spectrum. However, for the non-diagonal twist
matrices a natural reformulation of the transfer matrix spectrum leads to an
inhomogeneous functional equation. We have decided to develop the case of
non-diagonal twist matrix in some future analysis where, as already discussed
at the end of section 3.2, the main interesting question is if under an apropiate
choice of the $Q$-functions one can derive an homogeneous quantum spectral
curve characterization. Such a statement indeed holds for both the fundamental and
higher spin representations of the rank one trigonometric Yang-Baxter algebra, as proven in \cite{NicT15}.

\subsection{Fundamental properties of the hierarchy of transfer matrices}

Let us introduce the following operators, $\mathsf{N}_{i}\in $End$(\mathcal{H%
})$%
\begin{equation}
\mathsf{N}_{i}\otimes _{l=1}^{\mathsf{N}}|l,a_{l}\rangle =\otimes _{l=1}^{%
\mathsf{N}}|l,a_{l}\rangle \sum_{l=1}^{\mathsf{N}}\delta _{i,a_{l}}\text{ \ }%
\forall i\in \{1,...,n\},
\end{equation}%
in the basis%
\begin{equation}
|l,a_{l}\rangle =\left( 
\begin{array}{c}
\delta _{1,a_{l}} \\ 
\delta _{2,a_{l}} \\ 
\vdots \\ 
\delta _{n,a_{l}}%
\end{array}%
\right) _{l}\in V_{l}\text{ \ }\forall l\in \{1,...,\mathsf{N}\}\text{ and }%
a_{l}\in \{1,...,n\}\text{,}
\end{equation}%
so that it holds:%
\begin{equation}
\sum_{i=1}^{\mathsf{N}}\mathsf{N}_{i}=\mathsf{N}.
\end{equation}%
We will denote by $\nu_i$ the eigenvalues of the operators $\mathsf{N}_{i}$. Then we can recall the following relevant and known properties of the fused
transfer matrices:

\begin{proposition}
The higher transfer matrices:%
\begin{equation}
T_{m}^{\left( K\right) }(\lambda )\equiv \text{tr}%
_{a_{1},...,a_{m}}P_{a_{1},...,a_{m}}^{-}M_{a_{m}}^{(K)}(\lambda )\cdots
M_{a_{1}}^{(K)}(\lambda /q^{m-1}),\text{ \ \ }m\in \{1,...,n\}
\end{equation}%
defines one parameter families of mutually commuting operators: 
\begin{equation}
\left[ T_{i}^{\left( K\right) }(\lambda ),T_{j}^{\left( K\right) }(\mu )%
\right] =0,\text{ \ \ }i,j\in \{1,...,n\}.
\end{equation}%
The last quantum spectral invariant%
\begin{equation}
q\text{-det}M^{(K)}(\lambda )\equiv T_{n}^{\left( K\right) }(\lambda ),
\end{equation}%
the so-called quantum determinant, is a central element of the algebra with
the following explicit form:%
\begin{equation}
q\text{-det}M^{(K)}(\lambda )=\text{det}K\text{ }\prod_{b=1}^{\mathsf{N}%
}(\lambda q/\xi _{b}-\xi _{b}/(q\lambda ))\prod_{k=1}^{n-1}(\lambda
/(q^{k}\xi _{b})-(q^{k}\xi _{b})/\lambda ).
\end{equation}%
Moreover, the quantum spectral invariants have the following analyticity
properties:

a) The following fusion identities holds for any $a \in \{1, \dots, \mathsf{N}\}$:%
\begin{equation}
T_{1}^{\left( K\right) }(\xi _{a})T_{m-1}^{\left( K\right) }(\xi
_{a}/q)=T_{m}^{\left( K\right) }(\xi _{a}),\text{ \ }\forall m\in
\{2,...,n\},
\end{equation}%
and the following $(m-1)\mathsf{N}$ central zero conditions: 
\begin{equation}
\text{ }T_{m}^{\left( K\right) }(\pm q^{r}\xi _{a})=0\text{ \ }\forall r\in
\{1,...,m-1\},a\in \{1,...,\mathsf{N}\},
\end{equation}%
holds for any $m\in \{1,...,n-1\}$ and for any above diagonal $K$-matrix.

b) $\lambda ^{m\mathsf{N}}T_{m}^{\left( K\right) }(\lambda )$ is a degree $m%
\mathsf{N}$ polynomial in $\lambda ^{2}$ with the following asymptotics:%
\begin{equation}
T_{m}^{(\pm \infty |K)}\equiv \lim_{\log \lambda \rightarrow \pm \infty
}\lambda ^{\mp \mathsf{N}}T_{m}^{\left( K\right) }(\lambda )=\left( \pm
1\right) ^{m\mathsf{N}}\frac{\sigma _{m}^{(n)}(\mathsf{k}_{1}q^{\pm \mathsf{N%
}_{1}},...,\mathsf{k}_{n}q^{\pm \mathsf{N}_{n}})}{q^{\pm m(m-1)\mathsf{N}%
/2}\prod_{n=1}^{\mathsf{N}}\xi _{n}^{\pm m}},
\end{equation}%
in the case of the diagonal twist $K$, where $\sigma _{m}^{(n)}$ is the standard symmetric polynomial of degree $m$ in $n$ variables.
\end{proposition}

Let us introduce the functions%
\begin{align}
f_{l,\mathbf{h}}^{(m)}(\lambda )&=\frac{t_{h_{1},...,h_{\mathsf{N}}}\lambda
/\xi _{l}+\xi _{l}/(t_{h_{1},...,h_{\mathsf{N}}}\lambda )}{t_{h_{1},...,h_{%
\mathsf{N}}}+1/t_{h_{1},...,h_{\mathsf{N}}}}\prod_{b=1}^{\mathsf{N}%
}\prod_{r=1}^{m-1}\frac{\lambda /\xi _{b}^{(-r)}-\xi _{b}^{(-r)}/\lambda }{%
\xi _{l}^{(h_{l})}/\xi _{b}^{(-r)}-\xi _{b}^{(-r)}/\xi _{l}^{(h_{l})}} 
\notag \\
&\times\prod_{b\neq l,b=1}^{\mathsf{N}}\frac{\lambda /\xi _{b}^{(h_{b})}-\xi
_{b}^{(h_{b})}/\lambda }{\xi _{n}^{(h_{n})}/\xi _{b}^{(h_{b})}-\xi
_{b}^{(h_{b})}/\xi _{n}^{(h_{n})}}, \\
\end{align}
well defined under the assumption that $q$ is not a root of unity and $h_{l} \in \{0, \dots, n-1\}$ for any $l\in
\{1,...,\mathsf{N}\}$, and%
\begin{align}
T_{m,\mathbf{h}}^{(K,\infty )}(\lambda |\mathsf{N}_{1},...,\mathsf{N}_{n})& =%
\frac{\sum_{1\leq i_{1}<i_{2}<\cdots <i_{m-1}<i_{m}\leq n}\prod_{k=1}^{m}%
\mathsf{k}_{i_{k}}\cosh (\eta \sum_{k=1}^{m}\mathsf{N}_{i_{k}})}{\cosh
(\eta \sum_{a=1}^{\mathsf{N}}h_{a})}  \notag \\
& \times \prod_{b=1}^{\mathsf{N}}(\lambda /\xi _{b}^{(h_{b})}-\xi
_{b}^{(h_{b})}/\lambda )\prod_{r=1}^{m-1}(\lambda /\xi _{b}^{(-r)}-\xi
_{b}^{(-r)}/\lambda ),  \label{Asym-Tm}
\end{align}%
then the following corollary holds:

\begin{corollary}
Under the assumption that $q$ is not a root of unity and $h_{l} \in \{0, \dots, n-1\}$ for any $l\in
\{1,...,\mathsf{N}\}$, the following interpolation formulae:%
\begin{equation}
T_{m}^{(K)}(\lambda )=T_{m,\mathbf{h}}^{(K,\infty )}(\lambda |\mathsf{N}%
_{1},...,\mathsf{N}_{n})+\sum_{l=1}^{\mathsf{N}}f_{l,\mathbf{h}}^{\left(
m\right) }(\lambda )T_{m}^{(K)}(\xi _{l}^{(h_{l})})
\end{equation}%
holds for the transfer matrix $T_{m}^{(K)}(\lambda )$ with $m\in
\{1,...,n-1\}$ together with the following sum rules:%
\begin{align}
& \sum_{1\leq i_{1}<i_{2}<\cdots <i_{m-1}<i_{m}\leq n}\prod_{k=1}^{m}\mathsf{%
k}_{i_{k}}\cosh (\eta \sum_{k=1}^{m}\mathsf{N}_{i_{k}}-\sum_{l=1}^{\mathsf{N%
}}h_{l})  \notag \\
& =\sum_{l=1}^{\mathsf{N}}\frac{T_{m}^{(K)}(\xi _{l}^{\left( h_{l}\right) })%
}{2\prod_{b\neq l,b=1}^{\mathsf{N}}(\xi _{l}^{\left( h_{l}\right) }/\xi
_{b}^{\left( h_{b}\right) }-\xi _{b}^{\left( h_{b}\right) }/\xi _{l}^{\left(
h_{l}\right) })\prod_{b=1}^{\mathsf{N}}\prod_{r=1}^{m-1}(\xi _{l}^{\left(
h_{l}\right) }/\xi _{b}^{\left( -r\right) }-\xi _{b}^{\left( -r\right) }/\xi
_{l}^{\left( h_{l}\right) })}.
\end{align}%
The fusion equations allow to completely characterize all the $%
T_{m}^{(K)}(\lambda )$ in terms of $T_{1}^{(K)}(\lambda )$ by the following
interpolation formulae:%
\begin{equation}
T_{m}^{(K)}(\lambda )=T_{m,\mathbf{h=0}}^{(K,\infty )}(\lambda )+\sum_{l=1}^{%
\mathsf{N}}f_{n,\mathbf{h=0}}^{\left( m\right) }(\lambda )T_{m-1}^{(K)}(\xi
_{l}/q)T_{1}^{(K)}(\xi _{l}).
\end{equation}
\end{corollary}

\subsection{SoV covector basis generated by transfer matrix action}

The following theorem holds as a corollary of Proposition 2.6 of \cite%
{MaiN18}:

\begin{theorem}
Let $K$ be a $n\times n$ simple and diagonalizable symmetry of the $R$%
-matrix, then the following set:%
\begin{equation}
\langle h_{1},...,h_{\mathsf{N}}|\equiv \langle S|\prod_{n=1}^{\mathsf{N}%
}(T_{1}^{(K)}(\xi _{n}))^{h_{n}}\text{ \ for any }\{h_{1},...,h_{\mathsf{N}%
}\}\in \{0,...,n-1\}^{\otimes \mathsf{N}},
\end{equation}%
forms a covector basis of $\mathcal{H}$, for almost any choice of $\langle
S| $ and of the inhomogeneities satisfying $(\ref{Inhomog-cond})$. A proper
choice of the state $\langle S|$ has the following tensor product form:%
\begin{equation}
\langle S|=\bigotimes_{l=1}^{\mathsf{N}}(x_{1},...,x_{n})_{l}\Gamma
_{W}^{-1},\text{ \ \ }\Gamma _{W}=\bigotimes_{a=1}^{\mathsf{N}}W_{K,a}
\end{equation}%
under the condition $\prod_{l=1}^{n}x_{l}\neq 0$, where $W$ is the
invertible matrix defining the similarity transformation to the diagonal
matrix $K_{J}$ by $K=W_{K}K_{J}W_{K}^{-1}$.
\end{theorem}

\section{Transfer matrix spectrum in our SoV
approach: the $\mathcal{U}_{q}(\widehat{gl_{n}})$ case}

\subsection{Discrete spectrum characterization by SoV}

In the following we need $n-1$ Laurent polynomials in $\lambda $:%
\begin{equation}
t_{1}^{\left( K,\{x\},\{\nu_{i}\}\right) }(\lambda )=T_{m+1,\mathbf{h=0}%
}^{(K,\infty )}(\lambda |\nu_{1},...,\nu_{n})+\sum_{l=1}^{\mathsf{N}}f_{l,%
\mathbf{h=0}}^{\left( 1\right) }(\lambda )x_{l},  \label{Func-form-1}
\end{equation}%
where we have used the notation $T_{m+1,\mathbf{h=0}%
}^{(K,\infty )}(\lambda |\nu_{1},...,\nu_{n})$ for the eigenvalue of the already defined asymptotic operator $T_{m+1,\mathbf{h=0}%
}^{(K,\infty )}(\lambda |{\mathsf{N}}_{1},...,{\mathsf{N}}_{n})$ on the common eigenspaces of the operators ${\mathsf{N}}_i$  and 
\begin{equation}
t_{m+1}^{(K,\{x\},\{\nu_{i}\})}(\lambda )=T_{m+1,\mathbf{h=0}}^{(K,\infty
)}(\lambda |\nu_{1},...,\nu_{n})+\sum_{l=1}^{\mathsf{N}}f_{l,\mathbf{h=0}%
}^{\left( m+1\right) }(\lambda )t_{m}^{(K^{\left( a\right)
},\{x\},\{\nu_{i}\})}(\xi _{l}/q)x_{l},
\end{equation}%
for any $m\in \{1,...,n-2\}$, which are as well functions of a $n\times n$
twist matrix $K$, of a point $\{x_{1},...,x_{\mathsf{N}}\}\in \mathbb{C}^{%
\mathsf{N}}$ and of an $n$-upla $\{\nu_{1},...,\nu_{n}\}$ of nonnegative
integers (the eigenvalues of the operators $\mathsf{N}_i$) satisfying:%
\begin{equation}
\sum_{l=1}^{n}\nu_{l}=\mathsf{N}.
\end{equation}%
Then, the following characterization of the transfer matrix spectrum holds:

\begin{theorem}
\label{ch-discrete-U_q-n} We consider a twist $K$ matrix symmetry which is
diagonal and with simple spectrum and inhomogeneity parameters in generic position. Then
the spectrum of $T_{1}^{(K)}(\lambda )$ is characterized by:%
\begin{equation}
\Sigma _{T^{(K)}}=\bigcup_{\forall \nu_{i}\geq 0\text{ : }\sum_{l=1}^{n}\nu_{l}=%
\mathsf{N}}\Sigma _{T^{(K)}}^{\left( \{\nu_{i}\}\right) },
\end{equation}%
where%
\begin{equation}
\Sigma _{T^{(K)}}^{\left( \{\nu_{i}\}\right) }=\left\{ t_{1}(\lambda
):t_{1}(\lambda )=t_{1}^{\left( K,\{x\},\{\nu_{i}\}\right) }(\lambda ),\text{
\ }\forall \{x_{1},...,x_{\mathsf{N}}\}\in S_{T^{(K)}}^{\left(
\{\nu_{i}\}\right) }\right\} ,  \label{SET-T-n}
\end{equation}%
and $S_{T^{(K)}}^{\left( \{\nu_{i}\}\right) }$ is defined as the set of
solutions to the next system of $\mathsf{N}$ polynomial equations of order $%
n $:%
\begin{equation}
x_{a}t_{n-1}^{\left( K,\{x\},\{\nu_{i}\}\right) }(\xi _{a}/q)\left. =\right. q%
\text{-det}M^{(K)}(\xi _{a}),  \label{Ch-System-SoV-n}
\end{equation}%
in $\mathsf{N}$ unknown $\{x_{1},...,x_{\mathsf{N}}\}$. Moreover, $%
T_{1}^{(K)}(\lambda ,\{\xi \})$ is diagonalizable with simple spectrum and%
\begin{equation}
\langle h_{1},...,h_{\mathsf{N}}|t\rangle =\prod_{n=1}^{\mathsf{N}%
}t_{1}^{h_{n}}(\xi _{n})  \label{SoV-Ch-T-eigenV-n}
\end{equation}%
uniquely characterizes the eigenvector $|t\rangle $ associated to any fixed $%
t_{1}(\lambda )\in \Sigma _{T^{(K)}}$ in our SoV basis.
\end{theorem}

\begin{proof}
The proof works by some simple modifications of the case of the Yangian $%
Y(gl_{n})$ fundamental representations developed in our second paper \cite%
{MaiN18a}. We have just to handle the fact that the asymptotic behavior of
the transfer matrices is now not central in the full representation space
but only in each common eigenspaces of all the operators $\mathsf{N}_{i}$. Since those commute with all transfer matrices, the proof can be achieved in each of these subspaces and hence in the full Hilbert space.

In appendix we present an alternative proof of our statement which is a
corollary of the diagonalizabilty and simplicity of the transfer matrix
spectrum which follows from the Proposition 2.7 of our  paper \cite%
{MaiN18}.
\end{proof}

\subsection{Spectrum characterization by quantum spectral curve}

Let us first introduce the functions:%
\begin{equation}
\delta _{1}(\lambda )=\delta _{0}\prod_{a=1}^{\mathsf{N}}(\lambda q/\xi
_{a}-\xi _{a}/(\lambda q)),\text{ \ }\delta _{m}(\lambda
)=\prod_{i=0}^{m-1}\delta _{1}(\lambda /q^{i}),
\end{equation}
then the discrete SoV characterization of the transfer matrix spectrum
implies:

\begin{theorem}
Under the same conditions of the previous theorem, the entire functions $%
t_{1}(\lambda )$ is a $T_{1}^{\left( K\right) }(\lambda )$ transfer matrix
eigenvalue belonging to\footnote{%
i.e. the transfer matrix eigenvalues associated to the common eigenspace of
the $\mathsf{N}_{1},...,\mathsf{N}_{n}$ with eigenvalues $\nu_{1},...,\nu_{n}.$} 
$\Sigma _{T^{(K)}}^{\left( \{\nu_{i}\}\right) }$, with $\nu_{i}$ nonnegative
integers satisfying:%
\begin{equation}
\sum_{l=1}^{n}\nu_{l}=\mathsf{N},
\end{equation}%
iff there exists the unique Laurent polynomial:%
\begin{equation}
\varphi _{t}(\lambda )=\prod_{a=1}^{\mathsf{M}}(\lambda /\lambda
_{a}-\lambda _{a}/\lambda )\text{\ \ with }\mathsf{M}\leq \mathsf{N}\text{
and }\lambda _{a}\neq \xi _{m}\text{ }\forall (a,m)\in \{1,...,\mathsf{M}%
\}\times \{1,...,\mathsf{N}\},  \label{Phi-form-n}
\end{equation}%
such that $t_{1}(\lambda )$, $t_{m}(\lambda )\equiv t_{m}^{(K^{\left(
a\right) },\{t_{1}(\xi _{1}),...,t_{1}(\xi _{\mathsf{N}})\},\{\nu_{i}\})}(%
\lambda )$ and $\varphi _{t}(\lambda )$ are solutions of the following
quantum spectral curve functional equation:%
\begin{equation}
\sum_{l=0}^{n}\left( -1\right) ^{l}\delta _{l}(\lambda )\varphi _{t}(\lambda
/q^{l})t_{n-l}(\lambda /q^{l})\left. =\right. 0
\label{Q-Spectral-curve-U(gl_n)}
\end{equation}%
where $t_{0}(\lambda )=1$ and we have to fix\footnote{%
That is we have to fix $\delta _{0}$ to be one of the $n$ distinct
eigenvalue $\mathsf{k}_{i}$ of the twist matrix $K$ and then the degree of
the Laurent polynomial $\varphi _{t}(\lambda )$ is fixed by $\mathsf{M}=%
\mathsf{N}-\nu_{i}$.}:%
\begin{equation}
\delta _{0}=\mathsf{k}_{i}\text{ \ for one fixed }i\in \{1,...,n\},
\end{equation}%
and%
\begin{equation}
\mathsf{M}=\mathsf{N}-\nu_{i}.
\end{equation}%
Moreover, up to a normalization, the corresponding common transfer matrix
eigenvector $|t\rangle $ admits the following separate representation:%
\begin{equation}
\langle h_{1},...,h_{\mathsf{N}}|t\rangle =\prod_{a=1}^{\mathsf{N}}\delta
_{1}^{h_{a}}(\xi _{a})\varphi _{t}^{h_{a}}(\xi _{a}/q)\varphi
_{t}^{n-h_{a}}(\xi _{a}).
\end{equation}
\end{theorem}

\begin{proof}
Let us start proving that the asymptotics of the functional equation are
indeed compatibles with those of the transfer matrix eigenvalues. That is,
if we assume that $t_{1}(\lambda )\in \Sigma _{T^{(K)}}^{\left(
\{n_{i}\}\right) }$, then we have to show that the leading asymptotics
associated to the degree $\mathsf{M}+n\mathsf{N}$ of the l.h.s. of the
equation must be zero and vice versa.

Let us remark that from the known asymptotics $T_{m}^{(\pm \infty |K)}$ of
the transfer matrices, the following identities hold:%
\begin{equation}
\lim_{\log \lambda \rightarrow \pm \infty }\lambda ^{\mp (n-a)\mathsf{N}%
}T_{n-a}^{(K,\infty )}(\lambda /q^{a}|\mathsf{N}_{1},...,\mathsf{N}_{n})=%
\frac{\left( \pm 1\right) ^{(n-a)\mathsf{N}}\sigma _{n-a}^{(n)}(\mathsf{k}%
_{1}q^{\pm \mathsf{N}_{1}},...,\mathsf{k}_{n}q^{\pm \mathsf{N}_{n}})}{q^{\pm
(n-a)(a+(n-a-1)/2)\mathsf{N}}\prod_{n=1}^{\mathsf{N}}\xi _{n}^{\pm (n-a)}},
\end{equation}%
while it is easy to verify that it holds:%
\begin{equation}
\lim_{\log \lambda \rightarrow \pm \infty }\lambda ^{\mp a\mathsf{N}}\delta
_{a}(\lambda )=\frac{\left( \pm 1\right) ^{a\mathsf{N}}\mathsf{k}_{i}^{a}}{%
q^{\pm ((a-1)(a-2)/2-1)\mathsf{N}}\prod_{n=1}^{\mathsf{N}}\xi _{n}^{\pm a}},
\end{equation}%
where we have imposed the choice $\delta _{0}=$ $\mathsf{k}_{i}$ and%
\begin{equation}
\lim_{\log \lambda \rightarrow \pm \infty }\lambda ^{\mp \mathsf{M}}\varphi
_{t}(\lambda /q^{a})=\frac{\left( \pm 1\right) ^{\mathsf{M}}}{q^{\pm a%
\mathsf{M}}},
\end{equation}%
from which it follows:%
\begin{align}
\lim_{\log \lambda \rightarrow \pm \infty }\lambda ^{\mp (\mathsf{M}+n%
\mathsf{N})}(l.h.s.)_{\left( \ref{Q-Spectral-curve-U(gl_n)}\right) }& =\frac{%
\left( \pm 1\right) ^{(\mathsf{M}+n\mathsf{N})}\sum_{l=0}^{n}\left( -\mathsf{%
k}_{i}\right) ^{l}q^{\pm l(\mathsf{N}-\mathsf{M})}\sigma _{n-l}^{(n)}(%
\mathsf{k}_{1}q^{\pm \nu_{1}},...,\mathsf{k}_{n}q^{\pm \nu_{n}})}{q^{\pm n(n-1)%
\mathsf{N}/2}\prod_{n=1}^{\mathsf{N}}\xi _{n}^{\pm n}} \\
& =\frac{\left( \pm 1\right) ^{(\mathsf{M}+n\mathsf{N})}\sum_{l=0}^{n}\left(
-\mathsf{k}_{i}q^{\pm \nu_{i}}\right) ^{l}\sigma _{n-l}^{(n)}(\mathsf{k}%
_{1}q^{\pm \nu_{1}},...,\mathsf{k}_{n}q^{\pm \nu_{n}})}{q^{\pm n(n-1)\mathsf{N}%
/2}\prod_{n=1}^{\mathsf{N}}\xi _{n}^{\pm n}} \\
& =0,
\end{align}%
where according to our choice $\delta _{0}=\mathsf{k}_{i}$ we also fix $%
\mathsf{N}-\mathsf{M}=\nu_{i}$. The last identity to zero, holding for any
choice of $i\in \{1,..,n\}$, as a trivial consequence of the defining identity of the symmetric polynomials:%
\begin{equation}
\sum_{l=0}^{n}\left( -\lambda \right) ^{l}\sigma
_{n-l}^{(n)}(x_{1},...,x_{n})=\left( -1\right) ^{n}\prod_{a=1}^{n}(\lambda
-x_{a}),
\end{equation}%
which is zero if and only if $\lambda =x_{i}$ for any fixed $i\in \{1,..,n\}$%
. Vice versa, if $t_{1}(\lambda )$ satisfies with the polynomial $%
t_{m}(\lambda )$ and $\varphi _{t}(\lambda )$ the functional equation then
it is a degree $\mathsf{N}$ Laurent polynomial in $\lambda $ with leading
coefficients forced to be:%
\begin{equation}
t_{1}^{(\pm )}\equiv \lim_{\log \lambda \rightarrow \pm \infty }\lambda
^{\mp \mathsf{N}}t_{1}(\lambda )=\left( \pm 1\right) ^{\mathsf{N}}\frac{%
\sigma _{1}^{(n)}(\mathsf{k}_{1}q^{\pm n_{1}},...,\mathsf{k}_{n}q^{\pm
n_{n}})}{\prod_{n=1}^{\mathsf{N}}\xi _{n}^{\pm 1}},
\end{equation}%
as a consequence of the asymptotic of the satisfied functional equation for
any $\delta _{0}=\mathsf{k}_{i}$ we also fix $\mathsf{N}-\mathsf{M}=\nu_{i}$.

Now that the asymptotic behavior is verified the proof of the theorem follows mainly the same steps used for the Yangian $Y(gl_{n})$ case. For completeness let us
reproduce them here. We complete first the proof of the fact that given a $%
t_{1}(\lambda )$ entire function satisfying with the polynomials $%
t_{m}(\lambda )$ and $\varphi _{t}(\lambda )$ the functional equation
implies that it is a transfer matrix eigenvalue. Let us observe now that,
for $\lambda =\xi _{a}$ it holds:%
\begin{equation}
\delta _{1+j}(\xi _{a})=0,\text{ for }1\leq j\leq n-1,\text{ }\delta
_{1}(\xi _{a})\neq 0,\text{ det}_{q}M_{a}^{(K)}(\xi _{a})\neq 0,
\end{equation}%
and the quantum spectral curve in these points reads:%
\begin{equation}
\frac{\delta _{1}(\xi _{a})\varphi _{t}(\xi _{a}/q)}{\varphi _{t}(\xi _{a})}=%
\frac{\text{det}_{q}M_{a}^{(K)}(\xi _{a})}{t_{n-1}(\xi _{a}/q)}.
\label{Ch-1-Q-n}
\end{equation}%
Consider instead $1\leq s\leq n-1$, then for $\lambda =\xi _{a}q^{s}$ it holds:%
\begin{eqnarray}
\delta _{r\geq s+2}(\xi _{a}q^{s}) &=&0,\text{ }t_{n-b}(\xi _{a}q^{s-b})=0,%
\text{ for any }0\leq b\leq s-1 \\
\text{\ }\delta _{r\leq s+1}(\xi _{a}q^{s}) &\neq &0,
\end{eqnarray}%
and the quantum spectral curve in these points reads:%
\begin{equation}
\frac{\delta _{s+1}(\xi _{a}q^{s})\varphi _{t}(\xi _{a}/q)}{\delta _{s}(\xi
_{a}q^{s})\varphi _{t}(\xi _{a})}=\frac{t_{n-s}(\xi _{a})}{t_{n-s-1}(\xi
_{a}/q)}.  \label{Ch-2-Q-n}
\end{equation}%
Then the chain of identities:%
\begin{equation}
\frac{\delta _{s+1}(\xi _{a}q^{s})}{\delta _{s}(\xi _{a}q^{s})}=\delta
_{1}(\xi _{a})\text{ \ for any }1\leq s\leq n-1,
\end{equation}%
imply the following ones:%
\begin{equation}
t_{m+1}(\xi _{a})=t_{m}(\xi _{a}/q)t_{1}(\xi _{a}),\text{ }\forall m\in
\{1,...,n-1\},a\in \{1,...,\mathsf{N}\}.
\end{equation}%
So that $t_{m}(\lambda )$ are eigenvalues of the transfer matrices $%
T_{m}^{\left( K\right) }(\lambda )$, for the same eigenvector $|t\rangle $,
thanks to the SoV characterization given in our previous theorem.

Let us now prove the reverse statement. Let $t_{1}(\lambda )$ be eigenvalue
of the transfer matrix $T_{1}^{\left( K\right) }(\lambda )$ then we have to
prove the existence of $\varphi _{t}(\lambda )$ a Laurent polynomial which
satisfies the quantum spectral curve with the $t_{m}(\lambda )$. By imposing
the following set of conditions:%
\begin{equation}
\delta _{1}(\xi _{a})\frac{\varphi _{t}(\xi _{a}/q)}{\varphi _{t}(\xi _{a})}%
=t_{1}(\xi _{a}),  \label{Discrete-Ch-Q-n}
\end{equation}%
we characterize uniquely a Laurent polynomial $\varphi _{t}(\lambda )$ of
the form $\left( \ref{Phi-form}\right) $. Indeed, following the same steps
given in the proof of the Theorem 4.1 of our second paper \cite{MaiN18a} for
the Yangian $Y(gl_{n})$ case, we have that there exists a unique 
Laurent polynomial $\varphi _{t}(\lambda )$ of the form $\left( \ref%
{Phi-form}\right) $ with some degree $\mathsf{M}\leq \mathsf{N}$ so that one
is left with the proof of the identity $\mathsf{M}=\mathsf{N}-\nu_{i}$. This
is done just generalizing to the present case the argument based on the sum
rules as presented in the proof of the Theorem 4.3 of our first paper \cite%
{MaiN18}, see equations (4.69) and (4.70) there.

Here, we recall that this characterization of $\varphi _{t}(\lambda )$ indeed
implies that the functional equation is satisfied. The l.h.s. of the quantum
spectral curve is a Laurent polynomial in $\lambda $ of maximal degree $n%
\mathsf{N}+\mathsf{M}\leq (n+1)\mathsf{N}$ then to prove that it is
identically zero it is enough to show that it is zero in $(n+1)\mathsf{N}$
distinct points. Indeed, when this is the case the above argument on the sum
rules shows that the leading coefficients of the quantum spectral curve are
indeed zero. The chosen distinct points are the following $(n+1)\mathsf{N}$
ones $\xi _{a}q^{k_{a}}$, for any $a\in \{1,...,\mathsf{N}\}$ and $k_{a}\in
\left\{ -1,0,...,n-1\right\} $. For $\lambda =\xi _{a}/q$ we have:%
\begin{equation}
\delta _{r}(\xi _{a}/q)=0\text{ \ for any }1\leq r\leq n,\text{ as well as
det}M_{a}^{(K)}(\xi _{a}/q)=0,
\end{equation}%
so that for any $a\in \{1,...,\mathsf{N}\}$ the quantum spectral curve
equation is satisfied while in the remaining $n\mathsf{N}$ points this
equation coincides with the $n\mathsf{N}$ equations ($\ref{Ch-1-Q-n}$) and ($%
\ref{Ch-2-Q-n}$) which are satisfied by the transfer matrix eigenvalues as
they are all equivalent to the discrete characterization (\ref%
{Discrete-Ch-Q-n}) thanks to the fusion equations.

Let us verify now the equivalence of the SoV characterization of the
transfer matrix eigenvector with the one presented in this theorem. It is
enough to multiply by the non-zero product of the $\varphi _{t}^{n-1}(\xi
_{a})$ over all the $a\in \{1,...,\mathsf{N}\}$ the eigenvector $|t\rangle $
getting our result: 
\begin{equation}
\prod_{a=1}^{\mathsf{N}}\varphi _{t}^{n-1}(\xi _{a})\prod_{a=1}^{\mathsf{N}%
}t_{1}^{h_{a}}(\xi _{a})\overset{\left( \ref{Discrete-Ch-Q-n}\right) }{=}%
\prod_{a=1}^{\mathsf{N}}\delta _{1}^{h_{a}}(\xi _{a})\varphi
_{t}^{h_{a}}(\xi _{a}/q)\varphi _{t}^{n-1-h_{a}}(\xi _{a}).
\end{equation}
\end{proof}

\subsection{$Q$-operator reconstruction by SoV}

The $Q$-operator commuting family, satisfying the quantum spectral curve
equation with the transfer matrices at the operator level, can be
constructed in terms of the fundamental transfer matrix thanks to the above SoV characterization of the transfer matrix spectrum. Let us denote by $\delta _{b}^{(i)}(\lambda )$ the polynomials defined in the previous
section just making explicit that we have fixed $\delta _{0}=\mathsf{k}_{i}$%
\ for some fixed $i\in \{1,...,n\}$. Moreover, let us define the following $%
\mathsf{N}\times \mathsf{N}$ operator matrix of elements:%
\begin{equation}
\lbrack C_{i,\xi _{\mathsf{N}+1}}^{(T_{1}^{\left( K\right) })}]_{ab}=-\delta
_{ab}\,\frac{T_{1}^{\left( K\right) }(\xi _{a})}{\delta _{1}^{(i)}(\xi _{a})}%
+\prod_{\substack{ c=1  \\ c\neq b}}^{\mathsf{N}+1}\frac{\xi _{a}/(q\xi
_{c})-(q\xi _{c})/\xi _{a}}{\xi _{b}/\xi _{c}-\xi _{c}/\xi _{b}}\qquad
\forall a,b\in \{1,\ldots ,\mathsf{N}\},
\end{equation}%
and the rank one central matrix:%
\begin{equation}
\lbrack \Delta _{\xi _{\mathsf{N}+1}}(\lambda )]_{ab}=\frac{\lambda /\xi _{%
\mathsf{N}+1}-\xi _{\mathsf{N}+1}/\lambda }{\lambda /\xi _{b}-\xi
_{b}/\lambda }\frac{\prod_{c=1}^{\mathsf{N}}(\xi _{a}/(q\xi _{c})-(q\xi
_{c})/\xi _{a})}{\prod_{c=1,c\neq b}^{\mathsf{N}+1}(\xi _{b}/\xi _{c}-\xi
_{c}/\xi _{b})}\qquad \forall a,b\in \{1,\ldots ,\mathsf{N}\},
\label{Rank-one}
\end{equation}%
then it holds:

\begin{corollary}
Given\footnote{%
Note that we can fix for example $\xi _{N+1}=\xi _{h}-\eta $ for any fixed $%
h\in \{1,\ldots ,N\}$.} $\xi _{\mathsf{N}+1}\neq \xi _{i\leq N}$, then for
almost any values of the parameters $\{\xi _{i\leq N}\}$ and of the nonzero
eigenvalues $\{\mathsf{k}_{j\leq n}\}$ of the diagonal and simple spectrum
twist matrix $K$, the following Laurent polynomial family of $Q$-operators:%
\begin{equation}
Q_{_{i}}(\lambda )=\frac{\text{det}_{N}[C_{i,\xi _{\mathsf{N}%
+1}}^{(T_{1}^{\left( K\right) })}+\Delta _{\xi _{\mathsf{N}+1}}(\lambda )]}{%
\text{det}_{N}[C_{i,\xi _{\mathsf{N}+1}}^{(T_{1}^{\left( K\right) })}]}%
\prod_{c=1}^{\mathsf{N}}\frac{\lambda /\xi _{c}-\xi _{c}/\lambda }{\xi _{%
\mathsf{N}+1}/\xi _{c}-\xi _{\mathsf{N}+1}/\xi _{c}},
\end{equation}%
satisfies the operator quantum spectral curve equation%
\begin{equation}
\sum_{b=0}^{n}\delta _{b}^{(i)}(\lambda )Q_{i}(\lambda -b\eta
)T_{n-b}^{\left( K\right) }(\lambda -b\eta )=0,
\end{equation}%
where we have defined $T_{0}^{\left( K\right) }(\lambda )\equiv 1$, and
moreover $Q_{_{i}}(\xi _{a})$ are invertible operators for any $a\in
\{1,\ldots ,\mathsf{N}\}$.
\end{corollary}

\begin{proof}
The SoV characterization of the transfer matrix spectrum and the proof of
its reformulation in terms of the quantum spectral curve functional equation
imply this corollary. Indeed, following the same proof given in the case of
the fundamental representations of the Yangian $Y(gl_{n})$, see appendix B
of our second paper \cite{MaiN18a}, one can prove that the Laurent
polynomial $\varphi _{t}(\lambda )$ of the form $\left( \ref{Phi-form}%
\right) $ solution of the quantum spectral curve equation has the following
determinant representation:%
\begin{equation}
\varphi _{t}^{(i)}(\lambda )=\frac{\text{det}_{N}[C_{i,\xi _{\mathsf{N}%
+1}}^{(t_{1})}+\Delta _{\xi _{\mathsf{N}+1}}(\lambda )]}{\text{det}%
_{N}[C_{i,\xi _{\mathsf{N}+1}}^{(t_{1})}]}\prod_{c=1}^{\mathsf{N}}\frac{%
\lambda /\xi _{c}-\xi _{c}/\lambda }{\xi _{\mathsf{N}+1}/\xi _{c}-\xi _{%
\mathsf{N}+1}/\xi _{c}},
\end{equation}%
obtained by substituting to the transfer matrix $T_{1}^{\left( K\right)
}(\xi _{a})$ the corresponding eigenvalue $t_{1}(\xi _{a})$. As a
consequence of the Proposition 2.7 of \cite{MaiN18}, the transfer matrix $%
T_{1}^{\left( K\right) }(\lambda )$ is diagonalizable and with simple
spectrum in our current setting. The Laurent polynomial family $Q_{_{i}}$%
-operator is then completely characterized by its action on the eigenbasis
of the transfer matrix:%
\begin{equation}
Q_{_{i}}(\lambda )|t\rangle =|t\rangle \varphi _{t}^{(i)}(\lambda ),
\end{equation}%
for any eigenvalue $t_{1}(\lambda )$ and uniquely associated eigenvector $%
|t\rangle $ of the transfer matrix $T_{1}^{\left( K\right) }(\lambda )$,
which is equivalent to the characterization given in the corollary. This
also imply that this operator family satisfies the quantum spectral curve
equation with the transfer matrices.
\end{proof}

\section*{Acknowledgements}

J. M. M. and G. N. are supported by CNRS and ENS de Lyon.

\appendix

\section{Appendix A}

In this appendix we complete the proof of the Theorem \ref{Discrete-Ch-3_q}
by computing the direct computation of the action of transfer matrices on the
SoV covector basis. Once again let us comment that these computations are
obtained by adapting those of the fundamental representations of the $%
Y(gl_{3})$\ rational Yang-Baxter algebra in our first paper \cite{MaiN18},
taking into account the fact that  the transfer
matrices commute with the operators $\mathsf{N}_i$ that define their (non-central) asymptotic behavior. hence, all computations can be done in each common eigenspaces of these operators that give a complete decomposition of the full Hilbert space. 

\begin{proof}[Complement to proof of Theorem \protect\ref{Discrete-Ch-3_q}]
In order to complete the proof of Theorem \ref{Discrete-Ch-3_q}, we have to
prove that in the case $h_{j}=2$ the following identities hold 
\begin{equation}
\langle h_{1},...,h_{\mathsf{N}}|T_{1}^{(K^{\left( a\right) })}(\xi
_{j}/q)|t\rangle =t_{1}(\xi _{j}/q)\langle h_{1},...,h_{a},...,h_{\mathsf{N}%
}|t\rangle ,
\end{equation}%
and this is done by making an induction on the number $R$ of zeros contained
in $\{h_{1},...,h_{\mathsf{N}}\}\in \{0,1,2\}^{\otimes \mathsf{N}}$. Let us
start proving this identity for $R$ $=0$, the fusion identities imply:%
\begin{equation}
\langle h_{1},...,h_{a}=2,...,h_{\mathsf{N}}|T_{1}^{(K^{\left( a\right)
})}(\xi _{j}/q)|t\rangle =\langle h_{1},...,h_{j}=1,...,h_{\mathsf{N}%
}|T_{2}^{(K^{\left( a\right) })}(\xi _{j})|t\rangle ,
\end{equation}%
so that:%
\begin{eqnarray}
\langle h_{1},...,h_{j}^{\prime } &=&1,...,h_{\mathsf{N}}|T_{2}^{(K^{\left(
a\right) })}(\xi _{j})|t\rangle =\left. T_{2,\mathbf{h=1}}^{(K^{\left(
a\right) },\infty )}(\xi _{j})\right\vert _{\mathsf{N}_{1}=l,\mathsf{N}%
_{2}=m}\langle h_{1},...,h_{j}^{\prime },...,h_{\mathsf{N}}|t\rangle \\
&&+\sum_{n=1}^{\mathsf{N}}f_{n,\mathbf{h}=\mathbf{1}}^{(a,2)}(\xi
_{j})\langle h_{1},...,h_{j}^{\prime },...,h_{\mathsf{N}}|T_{2}^{\left(
K\right) }(\xi _{n}/q)|t\rangle ,
\end{eqnarray}%
thanks to the following interpolation formula:%
\begin{equation}
T_{2}^{\left( K\right) }(\xi _{j})=\left. T_{2,\mathbf{h=1}}^{(K^{\left(
a\right) },\infty )}(\xi _{j})\right\vert _{\mathsf{N}_{1}=l,\mathsf{N}%
_{2}=m}+\sum_{n=1}^{\mathsf{N}}f_{n,\mathbf{h}=\mathbf{1}}^{(a,2)}(\xi
_{j})T_{2}^{\left( K\right) }(\xi _{n}/q).
\end{equation}%
Then, being $R=0$, it follows:%
\begin{align}
\langle h_{1},...,h_{a}^{\prime },...,h_{\mathsf{N}}|T_{2}^{(K^{\left(
a\right) })}(\xi _{j})|t\rangle & =\left. T_{2,\mathbf{h=1}}^{(K^{\left(
a\right) },\infty )}(\xi _{j})\right\vert _{\mathsf{N}_{1}=l,\mathsf{N}%
_{2}=m}\langle h_{1},...,h_{j}^{\prime },...,h_{\mathsf{N}}|t\rangle  \notag
\\
& +\sum_{n=1}^{\mathsf{N}}q\text{-det}M^{(K^{\left( a\right) })}(\xi
_{n})f_{n,\mathbf{h}=\mathbf{1}}^{(a,2)}(\xi _{j})\langle
h_{1},...,h_{n}^{\prime \prime },...,h_{\mathsf{N}}|t\rangle ,
\end{align}%
where $h_{n}^{\prime \prime }=h_{n}-1$ for $n\neq j$ and $h_{j}^{\prime
\prime }=h_{j}^{\prime }-1=0$. Now the function:%
\begin{equation}
t_{2}(\lambda )=\left. T_{2,\mathbf{h=1}}^{(K^{\left( a\right) },\infty
)}(\lambda )\right\vert _{\mathsf{N}_{1}=l,\mathsf{N}_{2}=m}+\sum_{n=1}^{%
\mathsf{N}}f_{n,\mathbf{h}=\mathbf{0}}^{(a,2)}(\lambda )t_{1}(\xi
_{n}/q)t_{1}(\xi _{n}),  \label{function t2-def}
\end{equation}%
satisfies by definition the equations:%
\begin{eqnarray}
t_{2}(\xi _{n}) &=&t_{1}(\xi _{n}/q)t_{1}(\xi _{n}),\text{ }\forall n\in
\{1,...,\mathsf{N}\}, \\
t_{1}(\xi _{n})t_{2}(\xi _{n}/q) &=&q\text{-det}M^{(K^{\left( a\right)
})}(\xi _{n}),\text{ }\forall n\in \{1,...,\mathsf{N}\},
\end{eqnarray}%
where the the quantum determinant equation is indeed a consequence of the
definition of $t_{1}(\lambda )$. Then we get:%
\begin{align}
& \langle h_{1},...,h_{j}^{\prime },...,h_{\mathsf{N}}|T_{2}^{(K^{\left(
a\right) })}(\xi _{j})|t\rangle \left. =\right.  \notag \\
& =\left( \left. T_{2,\mathbf{h=1}}^{(K^{\left( a\right) },\infty )}(\xi
_{j})\right\vert _{\mathsf{N}_{1}=l,\mathsf{N}_{2}=m}+\sum_{n=1}^{\mathsf{N}%
}t_{2}(\xi _{n}/q)f_{n,\mathbf{h}=\mathbf{1}}^{(a,2)}(\xi _{j})\right)
\langle h_{1},...,h_{j}^{\prime },...,h_{\mathsf{N}}|t\rangle , \\
& =t_{2}(\xi _{n})\langle h_{1},...,h_{j}^{\prime },...,h_{\mathsf{N}%
}|t\rangle \\
& =t_{1}(\xi _{n}/q)\langle h_{1},...,h_{j}=2,...,h_{\mathsf{N}}|t\rangle ,
\end{align}%
where we have used the interpolation formula:%
\begin{equation}
t_{2}(\xi _{j})=\left. T_{2,\mathbf{h=1}}^{(K^{\left( a\right) },\infty
)}(\xi _{j})\right\vert _{\mathsf{N}_{1}=l,\mathsf{N}_{2}=m}+\sum_{n=1}^{%
\mathsf{N}}t_{2}(\xi _{n}/q)f_{n,\mathbf{h}=\mathbf{1}}^{(a,2)}(\xi _{j}),
\end{equation}%
i.e. we have shown our identity $(\ref{Id-step2})$ for $R=0$. Then we can do
our proof by induction; we assume that it holds for the generic $%
\{h_{1},...,h_{\mathsf{N}}\}\in \{0,1,2\}^{\otimes \mathsf{N}}$ containing $%
R-1$ zeros and we prove it for the generic $\{h_{1},...,h_{\mathsf{N}}\}\in
\{0,1,2\}^{\otimes \mathsf{N}}$ containing $R$ zeros. Let us fix the generic 
$\{h_{1},...,h_{\mathsf{N}}\}\in \{0,1,2\}^{\otimes \mathsf{N}}$ with $%
h_{j}=2$ and let us denote with $\pi $ a permutation of $\{1,...,\mathsf{N}%
\} $ such that:%
\begin{equation}
\left. 
\begin{array}{l}
h_{\pi (i)}=0,\text{ }\forall i\in \{1,...,R\}, \\ 
h_{\pi (i)}=1,\text{ }\forall i\in \{R+1,...,R+S\}, \\ 
h_{\pi (i)}=2,\text{ }\forall i\in \{R+S+1,...,\mathsf{N}\},%
\end{array}%
\right.
\end{equation}%
with $j=\pi (R+S+1)$. Let us use now the following interpolation formula:%
\begin{equation}
T_{2}^{\left( K\right) }(\xi _{j})=\left. T_{2,\mathbf{k}}^{(K^{\left(
a\right) },\infty )}(\xi _{j})\right\vert _{\mathsf{N}_{1}=l,\mathsf{N}%
_{2}=m}+\sum_{n=1}^{\mathsf{N}}f_{n,\mathbf{k}}^{(a,2)}(\xi
_{j})T_{2}^{\left( K\right) }(\xi _{n}^{\left( k_{n}\right) }),
\end{equation}%
where we have defined $\mathbf{k}$ by:%
\begin{equation}
\left. 
\begin{array}{l}
k_{\pi (i)}=1,\text{ }\forall i\in \{1,...,R\}, \\ 
h_{\pi (i)}=2,\text{ }\forall i\in \{R+1,...,\mathsf{N}\},%
\end{array}%
\right.
\end{equation}%
then it holds: 
\begin{align}
\langle h_{1},...,h_{j}^{\prime }& =1,...,h_{\mathsf{N}}|T_{2}^{(K^{\left(
a\right) })}(\xi _{j})|t\rangle =\left. T_{2,\mathbf{k}}^{(K^{\left(
a\right) },\infty )}(\xi _{j})\right\vert _{\mathsf{N}_{1}=l,\mathsf{N}%
_{2}=m}\langle h_{1},...,h_{j}^{\prime },...,h_{\mathsf{N}}|t\rangle  \notag
\\
& +\sum_{n=1}^{R}f_{\pi (n),\mathbf{k}}^{(a,2)}(\xi _{j})\langle
h_{1},...,h_{j}^{\prime },...,h_{\mathsf{N}}|T_{2}^{\left( K\right) }(\xi
_{\pi (n)})|t\rangle  \notag \\
& +\sum_{n=R+1}^{\mathsf{N}}f_{\pi (n),\mathbf{k}}^{(a,2)}(\xi _{j})\langle
h_{1},...,h_{j}^{\prime },...,h_{\mathsf{N}}|T_{2}^{\left( K\right) }(\xi
_{\pi (n)}/q)|t\rangle .
\end{align}%
and which by the fusion identity reads:%
\begin{align}
\langle h_{1},...,h_{j}^{\prime },...,h_{\mathsf{N}}|T_{2}^{(K^{\left(
a\right) })}(\xi _{j})|t\rangle & =\left. T_{2,\mathbf{k}}^{(K^{\left(
a\right) },\infty )}(\xi _{j})\right\vert _{\mathsf{N}_{1}=l,\mathsf{N}%
_{2}=m}\langle h_{1},...,h_{j}^{\prime },...,h_{\mathsf{N}}|t\rangle  \notag
\\
& +\sum_{n=1}^{R}f_{\pi (n),\mathbf{k}}^{(a,2)}(\xi _{j})\langle
h_{1}^{(n)},...,h_{\mathsf{N}}^{(n)}|T_{1}^{\left( K\right) }(\xi _{\pi
(n)}/q)|t\rangle  \notag \\
& +\sum_{n=R+1}^{\mathsf{N}}q\text{-det}M^{(K^{\left( a\right) })}(\xi _{\pi
(n)})f_{\pi (n),\mathbf{k}}^{(a,2)}(\xi _{j})\langle h_{1}^{(n)},...,h_{%
\mathsf{N}}^{(n)}|t\rangle ,  \label{T2-id-2}
\end{align}%
where we have defined:%
\begin{equation}
h_{\pi (m)}^{(n)}=\left\{ 
\begin{array}{l}
h_{\pi (m)}+\theta (R-m)\delta _{m,n}\text{ \ for }n\leq R \\ 
h_{\pi (m)}-\theta (m-(R+1))\delta _{m,n}-\delta _{m,R+S+1}\text{ \ for }%
R+1\leq n%
\end{array}%
\right. .
\end{equation}%
To compute $\langle h_{1}^{(n)},...,h_{\mathsf{N}}^{(n)}|T_{1}^{\left(
K\right) }(\xi _{\pi (n)}/q)|t\rangle $ for $n\leq R$, we use the following
interpolation formula:%
\begin{equation}
T_{1}^{\left( K\right) }(\xi _{\pi (n)}/q)=\left. T_{1,\mathbf{k}^{\prime
}}^{(K^{\left( a\right) },\infty )}(\xi _{\pi (n)}/q)\right\vert _{\mathsf{N}%
_{1}=l,\mathsf{N}_{2}=m}+\sum_{r=1}^{\mathsf{N}}f_{r,\mathbf{k}^{\prime
}}^{(a,1)}(\xi _{\pi (n)}/q)T_{1}^{\left( K\right) }(\xi
_{r}^{(k_{r}^{\prime })}),
\end{equation}%
where we have defined:%
\begin{equation}
k_{\pi (m)}^{\prime }=\left\{ 
\begin{array}{l}
0\text{ \ for }m\leq R+S+1 \\ 
1\text{ \ for }R+S+2\leq m%
\end{array}%
\right. ,
\end{equation}%
which gives:%
\begin{align}
\langle h_{1}^{(n)},...,h_{\mathsf{N}}^{(n)}|T_{1}^{\left( K\right) }(\xi
_{\pi (n)}/q)|t\rangle & =\left. T_{1,\mathbf{k}^{\prime }}^{(K^{\left(
a\right) },\infty )}(\xi _{\pi (n)}/q)\right\vert _{\mathsf{N}_{1}=l,\mathsf{%
N}_{2}=m}\langle h_{1}^{(n)},...,h_{\mathsf{N}}^{(n)}|t\rangle  \notag \\
& +\sum_{r=1}^{R+S+1}f_{\pi (r),\mathbf{k}^{\prime }}^{(a,1)}(\xi _{\pi
(n)}/q)\langle h_{1}^{(n)},...,h_{\mathsf{N}}^{(n)}|T_{1}^{\left( K\right)
}(\xi _{\pi (r)})|t\rangle  \notag \\
& +\sum_{r=R+S+2}^{\mathsf{N}}f_{\pi (r),\mathbf{k}^{\prime }}^{(a,1)}(\xi
_{\pi (n)}/q)\langle h_{1}^{(n)},...,h_{\mathsf{N}}^{(n)}|T_{1}^{\left(
K\right) }(\xi _{\pi (r)}/q)|t\rangle ,
\end{align}%
which becomes:%
\begin{align}
\langle h_{1}^{(n)},...,h_{\mathsf{N}}^{(n)}|T_{1}^{\left( K\right) }(\xi
_{\pi (n)}/q)|t\rangle & =\left. T_{1,\mathbf{k}^{\prime }}^{(K^{\left(
a\right) },\infty )}(\xi _{\pi (n)}/q)\right\vert _{\mathsf{N}_{1}=l,\mathsf{%
N}_{2}=m}\langle h_{1}^{(n)},...,h_{\mathsf{N}}^{(n)}|t\rangle  \notag \\
& +\sum_{r=1}^{R+S+1}f_{\pi (r),\mathbf{k}^{\prime }}^{(a,1)}(\xi _{\pi
(n)}/q)t_{1}(\xi _{\pi (r)})\langle h_{1}^{(n)},...,h_{\mathsf{N}%
}^{(n)}|t\rangle  \notag \\
& +\sum_{r=R+S+2}^{\mathsf{N}}f_{\pi (r),\mathbf{k}^{\prime }}^{(a,1)}(\xi
_{\pi (n)}/q)t_{1}(\xi _{\pi (r)}/q)\langle h_{1}^{(n)},...,h_{\mathsf{N}%
}^{(n)}|t\rangle ,
\end{align}%
where in the second line we have used the identity $(\ref{Id-step1})$ while
in the third line the identity $(\ref{Id-step2})$, which holds by assumption
being $R-1$ the number of zeros in $\{h_{1}^{(n)},...,h_{\mathsf{N}}^{(n)}\}$%
. So that we have shown for any $n\leq R$:%
\begin{equation}
\langle h_{1}^{(n)},...,h_{\mathsf{N}}^{(n)}|T_{1}^{\left( K\right) }(\xi
_{\pi (n)}/q)|t\rangle =t_{1}(\xi _{\pi (n)}/q)\langle h_{1}^{(n)},...,h_{%
\mathsf{N}}^{(n)}|t\rangle ,
\end{equation}%
and substituting it in the second line of (\ref{T2-id-2}), we get:%
\begin{align}
\langle h_{1},...,h_{j}^{\prime },...,h_{\mathsf{N}}|T_{2}^{(K^{\left(
a\right) })}(\xi _{j})|t\rangle & =\left. T_{2,\mathbf{k}}^{(K^{\left(
a\right) },\infty )}(\xi _{j})\right\vert _{\mathsf{N}_{1}=l,\mathsf{N}%
_{2}=m}\langle h_{1},...,h_{j}^{\prime },...,h_{\mathsf{N}}|t\rangle  \notag
\\
& +\sum_{n=1}^{R}t_{1}(\xi _{\pi (n)}/q)f_{\pi (n),\mathbf{k}}^{(a,2)}(\xi
_{j})\langle h_{1}^{(n)},...,h_{\mathsf{N}}^{(n)}|t\rangle  \notag \\
& +\sum_{n=R+1}^{\mathsf{N}}q\text{-det}M^{(K^{\left( a\right) })}(\xi _{\pi
(n)})f_{\pi (n),\mathbf{k}}^{(a,2)}(\xi _{j})\langle h_{1}^{(n)},...,h_{%
\mathsf{N}}^{(n)}|t\rangle ,
\end{align}%
and so $\langle h_{1},...,h_{j}^{\prime },...,h_{\mathsf{N}%
}|T_{2}^{(K^{\left( a\right) })}(\xi _{j})|t\rangle $ reads:%
\begin{align}
& \left( \left. T_{2,\mathbf{k}}^{(K^{\left( a\right) },\infty )}(\xi
_{j})\right\vert _{\mathsf{N}_{1}=l,\mathsf{N}_{2}=m}+\sum_{n=1}^{R}t_{1}(%
\xi _{\pi (n)})t_{1}(\xi _{\pi (n)}/q)f_{\pi (n),\mathbf{k}}^{(a,2)}(\xi
_{j})+\sum_{n=R+1}^{\mathsf{N}}t_{2}(\xi _{\pi (n)}/q)f_{\pi (n),\mathbf{k}%
}^{(a,2)}(\xi _{j})\right)  \notag \\
& \times \langle h_{1},...,h_{j}^{\prime },...,h_{\mathsf{N}}|t\rangle 
\notag \\
& =t_{2}(\xi _{j}/q)\langle h_{1},...,h_{j}^{\prime }=1,...,h_{\mathsf{N}%
}|t\rangle =t_{1}(\xi _{j}/q)\langle h_{1},...,h_{j}=2,...,h_{\mathsf{N}%
}|t\rangle ,
\end{align}%
i.e. we have proven our formula $(\ref{Id-step2})$. Finally, taken the
generic $\{h_{1},...,h_{\mathsf{N}}\}\in \{0,1,2\}^{\otimes \mathsf{N}}$
with:%
\begin{equation}
\left. 
\begin{array}{l}
h_{\pi (i)}=0,\text{ }\forall i\in \{1,...,R\}, \\ 
h_{\pi (i)}=1,\text{ }\forall i\in \{R+1,...,R+S\}, \\ 
h_{\pi (i)}=2,\text{ }\forall i\in \{R+S+1,...,\mathsf{N}\},%
\end{array}%
\right.
\end{equation}%
and by using the interpolation formula:%
\begin{equation}
T_{1}^{\left( K\right) }(\lambda )=\left. T_{1,\mathbf{k}}^{(K^{\left(
a\right) },\infty )}(\lambda )\right\vert _{\mathsf{N}_{1}=l,\mathsf{N}%
_{2}=m}+\sum_{n=1}^{\mathsf{N}}f_{n,\mathbf{p}}^{(a,1)}(\lambda
)T_{1}^{\left( K\right) }(\xi _{n}^{\left( p_{n}\right) }),
\end{equation}%
where we have defined $\mathbf{p}$ by:%
\begin{equation}
\left. 
\begin{array}{l}
p_{\pi (i)}=0,\text{ }\forall i\in \{1,...,R+S\}, \\ 
p_{\pi (i)}=1,\text{ }\forall i\in \{R+S+1,...,\mathsf{N}\},%
\end{array}%
\right.
\end{equation}%
then it holds: 
\begin{align}
\langle h_{1},...,h_{\mathsf{N}}|T_{1}^{\left( K\right) }(\lambda )|t\rangle
& =\left. T_{1,\mathbf{k}}^{(K^{\left( a\right) },\infty )}(\lambda
)\right\vert _{\mathsf{N}_{1}=l,\mathsf{N}_{2}=m}\langle h_{1},...,h_{%
\mathsf{N}}|t\rangle  \notag \\
& +\sum_{n=1}^{R}f_{\pi (n),\mathbf{p}}^{(a,1)}(\lambda )\langle
h_{1},...,h_{\mathsf{N}}|T_{1}^{\left( K\right) }(\xi _{\pi (n)})|t\rangle 
\notag \\
& +\sum_{n=R+1}^{\mathsf{N}}f_{\pi (n),\mathbf{p}}^{(a,1)}(\lambda )\langle
h_{1},...,h_{\mathsf{N}}|T_{1}^{\left( K\right) }(\xi _{\pi (n)}/q)|t\rangle
\end{align}%
then by using in the second line the identity $(\ref{Id-step1})$ and $(\ref%
{Id-step2})$ in the third line we get:%
\begin{align}
\langle h_{1},...,h_{\mathsf{N}}|T_{1}^{\left( K\right) }(\lambda )|t\rangle
& =\left( \left. T_{1,\mathbf{k}}^{(K^{\left( a\right) },\infty )}(\lambda
)\right\vert _{\mathsf{N}_{1}=l,\mathsf{N}_{2}=m}+\sum_{n=1}^{\mathsf{N}%
}f_{\pi (n),\mathbf{p}}^{(a,1)}(\lambda )t_{1}(\xi _{\pi (n)}^{(p_{\pi
(n)})})\right) \langle h_{1},...,h_{\mathsf{N}}|t\rangle \\
& =t_{1}(\lambda )\langle h_{1},...,h_{\mathsf{N}}|t\rangle,
\end{align}%
which complete the proof of our theorem.
\end{proof}

\section{Appendix B}

In this appendix, we provide a proof of the discrete SoV characterization of
the transfer matrix spectrum given in Theorem \ref{ch-discrete-U_q-n}
bypassing the computation of the transfer matrix action in the SoV basis. The proof is presented bellow in the case of the rational fundamental representations of 
$Y(\widehat{gl_{n}})$. Then one can either use the argument that the fundamental evaluation representations of $\mathcal{U}_{q}(\widehat{gl_{n}})$ lead
under the rational limit to the rational ones, so inferring that the same
result has to hold for the trigonometric case too for almost any values of
the parameters. Otherwise one can just repeat the same type of proof directly in
the trigonometric case, only taking into account that the asymptotic behavior for the
trigonometric case are not central in the full representation space but only
in the common eigenspaces of the operators $\mathsf{N}_{i}$.
Moreover, the case of non-fundamental representation can be handled
similarly. In fact, the proof of the Theorem 2.3 of our third paper \cite%
{MaiN18b} can be seen as the first step in the proof by induction for these
non-fundamental representations.

\begin{proof}[Proof of rational version of Theorem \protect\ref%
{ch-discrete-U_q-n}]
For the fundamental representations of $Y(gl_{n})$, the quantum separation
of variable characterization of the first transfer matrix spectrum reads%
\begin{equation}
\Sigma _{T^{(K)}}=\left\{ t_{1}(\lambda ):t_{1}(\lambda )=t_{1}^{\left(
K,\{x\}\right) }(\lambda ),\text{ \ }\forall \{x_{1},...,x_{\mathsf{N}}\}\in
S_{T}\right\} ,
\end{equation}%
in terms of the solutions to the following system $S_{T}$ of $\mathsf{N}$
polynomial equations of degree $n$:%
\begin{equation}
x_{a}t_{n-1}^{\left( K,\{x\}\right) }(\xi _{a}-\eta )\left. =\text{det}%
K\right. q\text{-det}M^{(I)}(\xi _{a}),
\end{equation}%
in $\mathsf{N}$ unknown $\{x_{1},...,x_{\mathsf{N}}\}$, where we recall the
definitions used in our second paper \cite{MaiN18}: 
\begin{equation}
t_{1}^{\left( K,\{x\}\right) }(\lambda )=\text{tr\thinspace }K\text{ }%
\prod_{a=1}^{\mathsf{N}}(\lambda -\xi _{a})+\sum_{a=1}^{\mathsf{N}}g_{a,%
\mathbf{h}=0}^{\left( 1\right) }(\lambda )x_{a},  \label{t1-form}
\end{equation}%
and:%
\begin{equation}
t_{m+1}^{\left( K,\{x\}\right) }(\lambda )=\prod_{b=1}^{\mathsf{N}%
}\prod_{r=1}^{m}(\lambda -\xi _{b}-r\eta )\left[ T_{m+1,\mathbf{h}=\mathbf{0}%
}^{(K,\infty )}(\lambda )+\sum_{a=1}^{\mathsf{N}}g_{a,\mathbf{h}=\mathbf{0}%
}^{\left( m+1\right) }(\lambda )x_{a}t_{m}^{\left( K,\{x\}\right) }(\xi
_{a}-\eta )\right] ,  \label{Rec-Func-form-m}
\end{equation}%
for any $m\in \{1,...,n-2\}$, and%
\begin{eqnarray}
T_{m,\mathbf{h}}^{(K,\infty )}(\lambda ) &=&\text{tr}_{1,...,m}\left[
P_{1,...,m}^{-}K_{1}K_{2}\cdots K_{m}\right] \prod_{b=1}^{\mathsf{N}%
}(\lambda -\xi _{b}^{(h_{n})}), \\
g_{a,\mathbf{h}}^{\left( m\right) }(\lambda ) &=&\prod_{b\neq a,b=1}^{%
\mathsf{N}}\frac{\lambda -\xi _{b}^{(h_{b})}}{\xi _{a}^{(h_{a})}-\xi
_{b}^{(h_{b})}}\prod_{b=1}^{\mathsf{N}}\prod_{r=1}^{m-1}\frac{1}{\xi
_{a}^{(h_{a})}-\xi _{b}^{(-r)}}.
\end{eqnarray}%
Here we are interested in giving a proof of this characterization bypassing
the computation of the action of the first transfer matrix in the SoV basis.
The fact that any eigenvalue defines a solutions of this system follows from the fusion relations. So the only nontrivial thing to show is that
indeed any solution of the above system defines one eigenvalue. The Theorem
of Be\`{z}out\footnote{See for example
William Fulton (1974). Algebraic Curves. Mathematics Lecture Note Series.
W.A. Benjamin.} states that the above system of polynomial equations admits $n^{\mathsf{N}}$
solutions if the $\mathsf{N}$ polynomials, defining the system, have no
common components\footnote{%
Indeed, if there are common components the system admits instead an infinite
number of solutions.}. The transfer matrix, being diagonalizable and
with simple spectrum, has exactly $n^{\mathsf{N}}$ distinct eigenvalues and
so, under the condition of no common components, there are indeed exactly $%
n^{\mathsf{N}}$ distinct solutions to the above system and each one is
uniquely associated to a transfer matrix eigenvalue.

We have to show now that the condition of no common components indeed holds
for almost any values of the parameters. The proof of this statement can be
done by induction on $n-1$ the rank of the Yang-Baxter algebra. Let us start with the
rank 1 case, i.e. $n=2$ and fundamental representations of $Y(gl_{2})$.
Here, we fix the eigenvalue of the twist matrix to be $\mathsf{k}_{1}\neq 0$
and $\mathsf{k}_{2}=0$, then the system of equations reads:%
\begin{equation}
t_{1}^{\left( K,\{x\}\right) }(\xi _{a})t_{1}^{\left( K,\{x\}\right) }(\xi
_{a}-\eta )=x_{a}t_{1}^{\left( K,\{x\}\right) }(\xi _{a}-\eta )\left. =\text{%
det}K\right. q\text{-det}M^{(I)}(\xi _{a})=0,  \label{Sys-n=2}
\end{equation}%
now taking into account that by definition $t_{1}^{\left( K,\{x\}\right)
}(\lambda )$ is a degree $\mathsf{N}$ polynomial in $\lambda $ and that it
holds:%
\begin{equation}
\xi _{a}^{(h)}\neq \xi _{b}^{(k)}\text{ \ \ }\forall h,k\in \{0,1\},a\neq
b\in \{1,...,\mathsf{N}\},
\end{equation}%
then a solution to the system can be obtained iff for any $a\in \{1,...,%
\mathsf{N}\}$ there exists a unique $h_{a}\in \{0,1\}$ such that $%
t_{1}^{\left( K,\{x\}\right) }(\xi _{a}^{(h_{a})})=0$, or equivalently:%
\begin{equation}
t_{1,\mathbf{h}}^{\left( K,\{x\}\right) }(\lambda )=\mathsf{k}%
_{1}\prod_{a=1}^{\mathsf{N}}(\lambda -\xi _{a}^{(h_{a})}).
\end{equation}%
So we have that the system has exactly $2^{\mathsf{N}}$ distinct solutions
associated to the $2^{\mathsf{N}}$ distinct $\mathsf{N}$-upla $\mathbf{h}%
=\{h_{1\leq n\leq \mathsf{N}}\}$ in $\bigotimes_{n=1}^{_{\mathsf{N}}}\{0,1\}$%
. So there are no common components for $\mathsf{k}_{1}\neq 0$ and $\mathsf{k%
}_{2}=0$, and being the polynomials defining the system (\ref{Sys-n=2}) also
polynomial in twist matrix eigenvalues we infer that this statement is true
for almost any choice of $\mathsf{k}_{1}$ and $\mathsf{k}_{2}$. So we have
proven our statement for $n=2$.

Let us now prove it for $n=3$, we fix here the twist matrix eigenvalues as
it follows $\mathsf{k}_{1}\neq 0$, $\mathsf{k}_{2}\neq 0$, $\mathsf{k}_{2}\neq  \mathsf{k}_{1}$ and $\mathsf{k}%
_{3}=0$, then the system of equations reads:%
\begin{equation}
x_{a}t_{2}^{\left( K,\{x\}\right) }(\xi _{a}-\eta )=\text{det}K\text{ }q%
\text{-det}M^{(I)}(\xi _{a})=0,  \label{Sys-n=3}
\end{equation}%
where by definition it holds%
\begin{equation}
x_{a}t_{1}^{\left( K,\{x\}\right) }(\xi _{a}-\eta )=t_{2}^{\left(
K,\{x\}\right) }(\xi _{a}),  \label{Fusion-1-1}
\end{equation}%
so that it holds too%
\begin{equation}
t_{2}^{\left( K,\{x\}\right) }(\xi _{a})t_{2}^{\left( K,\{x\}\right) }(\xi
_{a}-\eta )=0,  \label{Compl-Sys-n=3}
\end{equation}%
now taking into account that by definition $t_{2}^{\left( K,\{x\}\right)
}(\lambda )$ is a degree $2\mathsf{N}$ polynomial in $\lambda $, zero in the
points $\xi _{a}+\eta $ for any $a\in \{1,...,\mathsf{N}\}$, it follows that
a solution to the system (\ref{Compl-Sys-n=3}) can be obtained iff for any $%
a\in \{1,...,\mathsf{N}\}$ there exists a unique $h_{a}\in \{0,1\}$ such
that $t_{2}^{\left( K,\{x\}\right) }(\xi _{a}^{(h_{a})})=0$, or equivalently:%
\begin{equation}
t_{2,\mathbf{h}}^{\left( K,\{x\}\right) }(\lambda )=\mathsf{k}_{1}\mathsf{k}%
_{2}\prod_{a=1}^{\mathsf{N}}(\lambda -\xi _{a}-\eta )(\lambda -\xi
_{a}^{(h_{a})}).
\end{equation}%
So that the system (\ref{Compl-Sys-n=3}) has exactly $2^{\mathsf{N}}$
distinct solutions associated to the $2^{\mathsf{N}}$ distinct $\mathsf{N}$%
-upla $\mathbf{h}=\{h_{1\leq n\leq \mathsf{N}}\}$ in $\bigotimes_{n=1}^{_{%
\mathsf{N}}}\{0,1\}$. Now for any fixed $\mathbf{h}\in \bigotimes_{n=1}^{_{%
\mathsf{N}}}\{0,1\}$ we can define a permutation $\pi _{\mathbf{h}}\in S_{%
\mathsf{N}}$ and a nonnegative integer $m_{\mathbf{h}}\leq \mathsf{N}$ such
that:%
\begin{equation}
h_{\pi _{\mathbf{h}}(a)}=0\text{ \ }\forall a\in \{1,...,m_{\mathbf{h}}\}%
\text{ \ and \ }h_{\pi _{\mathbf{h}}(a)}=1\text{ \ }\forall a\in \{m_{%
\mathbf{h}}+1,...,\mathsf{N}\}\text{.}
\end{equation}%
It is easy to remark now that fixed $\mathbf{h}\in \bigotimes_{n=1}^{_{%
\mathsf{N}}}\{0,1\}$ then (\ref{Fusion-1-1}), for $a\in \{1,...,m_{\mathbf{h}%
}\}$, and (\ref{Sys-n=3}) are satisfied iff it holds:%
\begin{equation}
x_{\pi _{\mathbf{h}}(a)}=t_{1}^{\left( K,\{x\}\right) }(\xi _{\pi _{\mathbf{h%
}}(a)})=0\text{ \ }\forall a\in \{1,...,m_{\mathbf{h}}\}.
\label{Induced-zeros}
\end{equation}%
Indeed, if this is not the case for a given $b\in \{1,...,m_{\mathbf{h}}\}$,
then the (\ref{Sys-n=3}) implies $t_{2,\mathbf{h}}^{\left( K,\{x\}\right)
}(\xi _{h_{\pi _{\mathbf{h}}(b)}}-\eta )=0$ which is not compatible with our
choice of $t_{2,\mathbf{h}}^{\left( K,\{x\}\right) }(\lambda )$. So, for any
fixed $\mathbf{h}\in \bigotimes_{n=1}^{_{\mathsf{N}}}\{0,1\}$, we are left
with the requirement to satisfy the fusion equation (\ref{Fusion-1-1}) for $%
a\in \{m_{\mathbf{h}}+1,...,\mathsf{N}\}$ which results in the following
system of equation:%
\begin{equation}
t_{1}^{\left( K,\{x\},\mathbf{h}\right) }(\xi _{\pi _{\mathbf{h}%
}(a)})t_{1}^{\left( K,\{x\},\mathbf{h}\right) }(\xi _{\pi _{\mathbf{h}%
}(a)}-\eta )=t_{2,\mathbf{h}}^{\left( K,\{x\}\right) }(\xi _{\pi _{\mathbf{h}%
}(a)}),\text{ }\forall a\in \{m_{\mathbf{h}}+1,...,\mathsf{N}\},
\end{equation}%
where $t_{1}^{\left( K,\{x\},\mathbf{h}\right) }(\lambda )$ is a degree $%
\mathsf{N}$ polynomial in $\lambda $ of the form (\ref{t1-form}) with the $%
m_{\mathbf{h}}$ zeros given by (\ref{t1-form}). Then let us define the following
degree $\mathsf{N}-m_{\mathbf{h}}$ polynomial in $\lambda $:%
\begin{equation}
\bar{t}_{1}^{\left( K,\{x\},\mathbf{h}\right) }(\lambda )=t_{1}^{\left(
K,\{x\},\mathbf{h}\right) }(\lambda )/\prod_{a=1}^{m_{\mathbf{h}}}(\lambda
-\xi _{\pi _{\mathbf{h}}(a)}),
\end{equation}%
and the degree $2(\mathsf{N}-m_{\mathbf{h}})$ polynomial in $\lambda $:%
\begin{eqnarray}
\bar{t}_{2,\mathbf{h}}^{\left( K,\{x\}\right) }(\lambda ) &=&t_{2,,\mathbf{h}%
}^{\left( K,\{x\}\right) }(\lambda )/\prod_{a=1}^{m_{\mathbf{h}}}\left[
(\lambda -\xi _{\pi _{\mathbf{h}}(a)})(\lambda -\xi _{\pi _{\mathbf{h}%
}(a)}-\eta )\right] \\
&=&\mathsf{k}_{1}\mathsf{k}_{2}\prod_{a=1+m_{\mathbf{h}}}^{\mathsf{N}%
}(\lambda -\xi _{\pi _{\mathbf{h}}(a)}-\eta )(\lambda -\xi _{\pi _{\mathbf{h}%
}(a)}+\eta ),
\end{eqnarray}
the previous system of equations simplifies to:%
\begin{equation}
\bar{t}_{1}^{\left( K,\{x\},\mathbf{h}\right) }(\xi _{\pi _{\mathbf{h}}(a)})%
\bar{t}_{1}^{\left( K,\{x\},\mathbf{h}\right) }(\xi _{\pi _{\mathbf{h}%
}(a)}-\eta )=\bar{t}_{2,\mathbf{h}}^{\left( K,\{x\}\right) }(\xi _{\pi _{%
\mathbf{h}}(a)}),\text{ }\forall a\in \{m_{\mathbf{h}}+1,...,\mathsf{N}\}.
\end{equation}%
Such a system coincides with the system associated to the case $n=2$ for a
lattice with $\mathsf{N}-m_{\mathbf{h}}$ sites and inhomogeneities $\xi
_{\pi _{\mathbf{h}}(a)}$ with $a\in \{m_{\mathbf{h}}+1,...,\mathsf{N}\}$.
Indeed, $\bar{t}_{2,\mathbf{h}}^{\left( K,\{x\}\right) }(\lambda )$ is just
the quantum determinant for such a lattice associated to the $2\times 2$
twist matrix with distinct non-zero eigenvalues $\mathsf{k}_{1}\neq 0$ and $\mathsf{k}_{2}\neq
0$ and $\bar{t}_{1}^{\left( K,\{x\},\mathbf{h}\right) }(\lambda )$ has the
functional form of a transfer matrix eigenvalue with asymptotic given by the
trace of this twist matrix. Now, we can use our result for $n=2$ to state
that this system has exactly 2$^{\mathsf{N}-m_{\mathbf{h}}}$ distinct
solutions, which allows to count the full set of solutions to our original
system:%
\begin{equation}
\sum_{m=0}^{\mathsf{N}}2^{\mathsf{N}-m}\left( 
\begin{array}{c}
\mathsf{N} \\ 
m%
\end{array}%
\right) =3^{\mathsf{N}},
\end{equation}
where we have used that for any fixed $m\in \{1,...,\mathsf{N}\}$ the number
of $\mathbf{h}\in \bigotimes_{n=1}^{_{\mathsf{N}}}\{0,1\}$ such that $m_{%
\mathbf{h}}=m$ is exactly given by the binomial symbol:%
\begin{equation}  \label{binomial}
\left( 
\begin{array}{c}
\mathsf{N} \\ 
m%
\end{array}%
\right) =\frac{\mathsf{N}!}{(\mathsf{N}-m)!m!}.
\end{equation}%
So we proved that the system has exactly $3^{\mathsf{N}}$ distinct solutions
and no common components for $\mathsf{k}_{1}\neq 0$, $\mathsf{k}_{2}\neq 0$
and $\mathsf{k}_{3}=0$, and being the polynomials defining the system (\ref%
{Sys-n=3}) also polynomial in twist matrix eigenvalues we can infer that this
statement is true for almost any choice of three distinct eigenvalues $\mathsf{k}_{1}$, $\mathsf{k}_{2}$
and $\mathsf{k}_{3}$. So we have proven our statement for $n=3$.

At this point it is easy to understand how to implement the proof by
induction, i.e. we assume that the statement is proven for the rank $n-1$
case and we prove it for the rank $n$ case and this is done in the case of a
diagonalizable and simple spectrum $(n+1)\times (n+1)$ twist matrix with pairwise distinct 
eigenvalues $\mathsf{k}_{a}\neq 0$, for any $a\in \{1,...,n\}$, and $\mathsf{%
k}_{n+1}=0$. Then following similar steps to those illustrated above, we see
that the function $t_{n}^{\left( K,\{x\}\right) }(\lambda )$ is forced to
take the form%
\begin{equation}
t_{n,\mathbf{h}}^{\left( K,\{x\}\right) }(\lambda )=\prod_{a=1}^{\mathsf{N}}%
\mathsf{k}_{a}(\lambda -\xi _{a}^{(h_{a})})\prod_{r=1}^{n-1}(\lambda -\xi
_{a}-r\eta ),
\end{equation}%
associated to the $2^{\mathsf{N}}$ distinct $\mathsf{N}$-upla $\mathbf{h}%
=\{h_{1\leq n\leq \mathsf{N}}\}$ in $\bigotimes_{n=1}^{_{\mathsf{N}}}\{0,1\}$
and that for any fixed $\mathbf{h}$ the system is reduced to that associated
to the case of rank $n-1$ with general diagonalizable and simple spectrum $%
n\times n$ twist matrix with eigenvalues $\mathsf{k}_{a}\neq 0$, for any $%
a\in \{1,...,n\}$, on a number of site $\mathsf{N}-m_{\mathbf{h}}$. Then
using the induction we know that this system admits $n^{\mathsf{N}-m_{%
\mathbf{h}}}$ distinct solutions for any such $\mathbf{h}$ and that for any
fixed $m\in \{1,...,\mathsf{N}\}$ the number of $\mathbf{h}\in
\bigotimes_{n=1}^{_{\mathsf{N}}}\{0,1\}$ such that $m_{\mathbf{h}}=m$ is
exactly given by the binomial symbol (\ref{binomial}). So that the total
counting gives:%
\begin{equation}
\sum_{m=0}^{\mathsf{N}}n^{\mathsf{N}-m}\left( 
\begin{array}{c}
\mathsf{N} \\ 
m%
\end{array}%
\right) =(n+1)^{\mathsf{N}},
\end{equation}%
which proves the no common component statement for the rank $n$ case too
when we repeat the polynomiality argument of the dependence w.r.t. the twist
matrix eigenvalues.
\end{proof}


\begin{thebibliography}{99}
\bibitem{MaiN18} J.~M. Maillet, and G.~Niccoli. \newblock On quantum
separation of variables \newblock {\em J. Math. Phys.} {\bf 59} (2018)
091417.

\bibitem{MaiN18a} J.~M. Maillet, and G.~Niccoli. \newblock Complete spectrum
of quantum integrable lattice models associated to $Y(gl_n)$ by separation
of variables \newblock {\em arXiv:1810.11885}.

\bibitem{MaiN18b} J.~M. Maillet, and G.~Niccoli. \newblock On quantum
separation of variables: beyond fundamental representations, to appear.

\bibitem{FadS78} L.~D. Faddeev and E.~K. Sklyanin. \newblock %
Quantum-mechanical approach to completely integrable field theory models. %
\newblock {\em Sov. Phys. Dokl.}, 23:902--904, 1978.

\bibitem{FadST79} L.~D. Faddeev, E.~K. Sklyanin, and L.~A. Takhtajan. %
\newblock Quantum inverse problem method {I}. 
\newblock {\em Theor. Math.
Phys.}, 40:688--706, 1979. \newblock Translated from Teor. Mat. Fiz. 40
(1979) 194-220.

\bibitem{FadT79} L.~A. Takhtadzhan and L.~D. Faddeev. \newblock The quantum
method of the inverse problem and the {H}eisenberg {XYZ} model. \newblock 
\emph{Russ. Math. Surveys}, 34(5):11--68, 1979.

\bibitem{Skl79} E.~K. Sklyanin. \newblock Method of the inverse scattering
problem and the non-linear quantum {S}chr{\"o}dinger equation. \newblock 
\emph{Sov. Phys. Dokl.}, 24:107--109, 1979.

\bibitem{Skl79a} E.~K. Sklyanin. \newblock On complete integrability of the {%
L}andau-{L}ifshitz equation. \newblock Preprint LOMI E-3-79, 1979.

\bibitem{FadT81} L.~D. Faddeev and L.~A. Takhtajan. \newblock Quantum
inverse scattering method. \newblock {\em Sov. Sci. Rev. Math.}, C 1:107,
1981.

\bibitem{Skl82} E.~K. Sklyanin. \newblock Quantum version of the inverse
scattering problem method. \newblock {\em J. Sov. Math.}, 19:1546--1595,
1982.

\bibitem{Fad82} L.~D. Faddeev. \newblock Integrable models in $(1 + 1)$%
-dimensional quantum field theory. \newblock In J.~B. Zuber and R.~Stora,
editors, \emph{Les Houches 1982, Recent advances in field theory and
statistical mechanics}, pages 561--608. Elsevier Science Publ., 1984.

\bibitem{Fad96} L.~D. Faddeev. \newblock How algebraic {B}ethe ansatz works
for integrable model. \newblock Les Houches Lectures, 1996.

\bibitem{BabdeVV81} O. Babelon, H.J. de Vega and C.M. Viallet, \newblock %
Solutions of the factorization equations from Toda field theory\newblock 
\emph{Nucl. Phys. B}, 190 [FS3] (1981) 542-552.

\bibitem{Jim86} M.~Jimbo, \newblock Quantization of solitons and the
restricted sine-gordon model. 
\newblock {\em Comm. Math. 
Phys.}, 102: 537--547 (1996).

\bibitem{KulR83a} {P. P. Kulish and N. Yu. Reshetikhin}. 
\newblock {Quantum linear problem for the sine-Gordon equation and higher
  representations}. \newblock {\em {J. Sov. Math.}}, {23}:{%
2435--41}, {1983}.

\bibitem{KulR82} P.~P. Kulish and N.~{Yu}. Reshetikhin. \newblock {$Gl_3$}%
-invariant solutions of the {Y}ang-{B}axter equation and associated quantum
systems. \newblock {\em J. Sov. Math.}, 34:1948--1971, 1986. \newblock %
translated from Zap. Nauch. Sem. LOMI 120, 92-121 (1982).

\bibitem{Jim85} {M. Jimbo}. 
\newblock {A q-difference analogue of U(g) and
the Yang-Baxter equation}. \newblock {\em {Lett. Math. Phys.}}%
, {10}:{63-69}, {1985}.

\bibitem{Dri87} {V. G. Drinfel'd}. 
\newblock {Quantum groups}. 
\newblock {\em {Proceedings of the International Congress of Mathematicians
  Berkeley, USA}}, {10}:{798-820}, {1986}.

\bibitem{ChaP94} {V. Chari and A. Pressley,}. 
\newblock {A Guide to Quantum
Groups}. \newblock {\em {Cambridge University Press, Cambridge.}}, (1994).

\bibitem{KulRS81} P.~P. Kulish, N.~Yu. Reshetikhin, and E.~K. Sklyanin. %
\newblock {Y}ang-{B}axter equation and representation theory {I}. \newblock 
\emph{Lett. Math. Phys.}, 5:393--403, 1981.

\bibitem{KulR83} P. ~P. Kulish, N.~Yu. Reshetikhin, \newblock %
Diagonalization of GL(N) invariant transfer matrices and quantum N-wave
system (Lee model), \newblock {\em J. Phys. A} {\bf 16} (1983) L591.

\bibitem{Res83} N.~Yu. Reshetikhin, \newblock A method of functional
equations in the theory of exactly solvable quantum systems 
\newblock {\em
Lett. Math. Phys.} {\bf 7} (1983) 205.

\bibitem{Res83a} N.~Yu. Reshetikhin, \newblock The functional equation
method in the theory of exactly soluble quantum systems 
\newblock {\em Sov.
Phys. JETP } {\bf 57} (1983) 691.

\bibitem{KirR86} A. ~N. Kirillov and N.~Yu. Reshetikhin, \newblock Yangians,
Bethe ansatz and combinatorics \newblock {\em Lett. Math. Phys.} {\bf 12}
(1986) 199.

\bibitem{Res87} N.~Yu. Reshetikhin, \newblock The spectrum of the transfer
matrices connected with Kac-Moody algebras \newblock {\em Lett. Math. Phys.} 
\textbf{14} (1987) 235.

\bibitem{KunNS94} A.~Kuniba, T.~Nakanishi, J.~Suzuki, \newblock Functional
relations in solvable lattice models I: Functional relations and
representation theory, \newblock {\em Int. J. Mod. Phys.} {\bf A9} (1994)
5215.

\bibitem{KulR81} P. ~P. Kulish, N.~Yu. Reshetikhin, \newblock Generalized
Heisenberg ferromagnet and the Gross-Neveu model, \newblock 
\emph{Sov. Phys. JETP}, \textbf{53} (1981) 108.

\bibitem{BelR08} S. Belliard and E. Ragoucy \newblock The nested Bethe
ansatz for 'all' closed spin chains \newblock {\em J. Phys. A}
(2008) \textbf{41} 295202.

\bibitem{PakRS18} S. Pakuliak, E. Ragoucy and N. Slavnov, \newblock Nested
Algebraic Bethe Ansatz in integrable models: recent results 
\newblock {\em
SciPost Phys. Lect. Notes} (2018) \textbf{6}.

\bibitem{BelPRS12} S. Belliard, S. Pakuliak, E. Ragoucy and N.~A. Slavnov %
\newblock Highest coefficient of scalar products in SU(3)-invariant
integrable models \newblock {\em J. Stat. Phys.} (2012) P09003.

\bibitem{BelPRS12a} S. Belliard, S. Pakuliak, E. Ragoucy and N.~A. Slavnov %
\newblock The algebraic Bethe ansatz for scalar products in SU (3)-invariant
integrable models \newblock {\em J. Stat. Phys.} (2012) P10017.

\bibitem{BelPRS13} S. Belliard, S. Pakuliak, E. Ragoucy and N.~A. Slavnov %
\newblock Bethe vectors of GL(3)-invariant integrable models, 
\newblock {\em
J. Stat. Phys.} (2013) P02020.

\bibitem{BelPRS13a} S. Belliard, S. Pakuliak, E. Ragoucy and N.~A. Slavnov %
\newblock  Form factors in SU(3)-invariant integrable models, 
\newblock {\em
J. Stat. Phys.} (2013) P04033.

\bibitem{PakRS14} S. Pakuliak, E. Ragoucy and N.~A. Slavnov \newblock   Form
factors in quantum integrable models with GL(3)-invariant R-matrix, %
\newblock {\em Nucl. Phys. B} {\bf 881} (2014) 343.

\bibitem{PakRS14a} S. Pakuliak, E. Ragoucy and N.~A. Slavnov \newblock   %
Determinant representations for form factors in quantum integrable models
with the GL(3)-invariant R-matrix, 
\newblock {\em Theor. Math. Phys.} {\bf
181} (2014) 1566.

\bibitem{PakRS15} S. Pakuliak, E. Ragoucy and N.~A. Slavnov \newblock     %
GL(3)-Based Quantum Integrable Composite Models. I. Bethe Vectors, \newblock 
\emph{SYGMA} \textbf{11} (2015) 063.

\bibitem{PakRS15a} S. Pakuliak, E. Ragoucy and N.~A. Slavnov \newblock    %
GL(3)-Based Quantum Integrable Composite Models. II. Form Factors of Local
Operators \newblock {\em SYGMA} {\bf 11} (2015) 064.

\bibitem{PakRS15b} S. Pakuliak, E. Ragoucy and N.~A. Slavnov \newblock    %
Form factors of local operators in a one-dimensional two-component Bose gas, %
\newblock {\em J. Phys. A} {\bf 48} (2015) 435001.

\bibitem{LiaS18} A.~Liashyk and N.~A. Slavnov. \newblock On {B}ethe vectors
in gl3 -invariant integrable models. \newblock {\em JHEP}, 06:18, 2018.

\bibitem{HaoCLYSW16} K. Hao, J. Cao, G.-L. Li, W.-L. Yang, K. Shi and Y.
Wang, \newblock Exact solution of an su(n) spin torus 
\newblock {\em J.
Stat. Mech.} (2016) 073104.

\bibitem{Skl85} E.~K. Sklyanin. \newblock The quantum {T}oda chain. %
\newblock In N.~Sanchez, editor, \emph{Non-Linear Equations in Classical and
Quantum Field Theory}, pages 196--233. Springer Berlin Heidelberg, 1985.

\bibitem{Skl90} E.~K. Sklyanin. \newblock Functional {B}ethe {A}nsatz. %
\newblock In B.A. Kupershmidt, editor, \emph{Integrable and Superintegrable
Systems}, pages 8--33. World Scientific, Singapore, 1990.

\bibitem{Skl92} E.~K. Sklyanin. \newblock Quantum inverse scattering method. 
{S}elected topics. \newblock In Mo-Lin Ge, editor, \emph{Quantum Group and
Quantum Integrable Systems}, pages 63--97. Nankai Lectures in Mathematical
Physics, World Scientific, 1992.

\bibitem{Skl92a} E.~K. Sklyanin. \newblock Separation of variables in the
classical integrable sl3 magnetic chain. \newblock {\em Comm. Math. Phys.},
150:181--191, 1992.

\bibitem{Skl95} E.~K. Sklyanin. \newblock Separation of variables. {N}ew
trends. \newblock {\em Prog. Theor. Phys. Suppl.}, 118:35--60, 1995.

\bibitem{Skl96} E.~K. Sklyanin. \newblock Separation of variables in the
quantum integrable models related to the yangian [sl(3)]. 
\newblock {\em J.
Math. Sci.}, 80:1861--1871, 1996.

\bibitem{BabBS96} O.~Babelon, D.~Bernard, and F.~A. Smirnov. \newblock %
Quantization of solitons and the restricted sine-gordon model. \newblock 
\emph{Comm. Math. Phys.}, 182(2):319--354, Dec 1996.

\bibitem{Smi98a} F.~A. Smirnov. \newblock Structure of matrix elements in
the quantum {T}oda chain. \newblock {\em J. Phys. A},
31(44):8953, 1998.

\bibitem{Smi01} F.~Smirnov. 
\newblock {Separation of variables for quantum integrable models related to
  Uq(slN)}. \newblock {\em arXiv:math-ph-010901}, 2001.

\bibitem{DerKM01} S.~E. Derkachov, G.P. Korchemsky, and A.~N. Manashov. 
\newblock {Noncompact Heisenberg spin magnets from high--energy QCD. I. Baxter
  Q--operator and separation of variables}. \newblock {\em Nucl. Phys. B},
617:375--440, 2001.

\bibitem{DerKM03} S.~E. Derkachov, G.~P. Korchemsky, , and A.~N. Manashov. %
\newblock Separation of variables for the quantum {SL(2,$\mathbb{\ R}$)}
spin chain. \newblock {\em JHEP}, 07:047, 2003.

\bibitem{DerKM03b} S.~E. Derkachov, G.~P. Korchemsky, and A.~N. Manashov. %
\newblock Baxter {Q}-operator and separation of variables for the open {SL(2,%
$\mathbb{R}$)} spin chain. \newblock {\em JHEP}, 10:053, 2003.

\bibitem{BytT06} A.~Bytsko and J.~Teschner. 
\newblock {Quantization of models with non--compact quantum group symmetry.
  Modular XXZ magnet and lattice sinh--Gordon model}. 
\newblock {\em J.
Phys. A}, 39:12927, 2006.

\bibitem{vonGIPS06} G.~von Gehlen, N.~Iorgov, S.~Pakuliak, and V.~Shadura. 
\newblock {The Baxter--Bazhanov--Stroganov model: separation of variables and
  the Baxter equation}. \newblock {\em J. Phys. A}, 39:7257,
2006.

\bibitem{FraSW08} H.~Frahm, A.~Seel, and T.~Wirth. \newblock Separation of
variables in the open {XXX} chain. \newblock {\em Nucl. Phys. B},
802:351--367, 2008.

\bibitem{MukTV09} E.~Mukhin, V.~Tarasov, and A.~Varchenko. \newblock On
separation of variables and completeness of the {B}ethe ansatz for quantum
gln {G}audin model. \newblock {\em Glasgow Math. J.}, 51:137--145, 2009.

\bibitem{MukTV09a} E.~Mukhin, V.~Tarasov, and A.~Varchenko. \newblock %
Schubert calculus and the representations of the general linear group. %
\newblock {\em J. Am. Math. Soc.}, 22:909--940, 2009.

\bibitem{MukTV09c} E.~Mukhin, V.~Tarasov, and A.~Varchenko. \newblock The {B}%
. and {M}. {S}hapiro conjecture in real algebraic geometry and the {B}ethe
ansatz. \newblock {\em Ann. Math.}, 170:863--881, 2009.

\bibitem{AmiFOW10} L.~Amico, H.~Frahm, A.~Osterloh, and T.~Wirth. \newblock %
Separation of variables for integrable spin-boson models. 
\newblock {\em
Nucl. Phys. B}, 839(3):604 -- 626, 2010.

\bibitem{NicT10} G.~Niccoli and J.~Teschner. 
\newblock {The Sine-Gordon
model revisited I}. \newblock {\em J. Stat. Mech.}, 2010:P09014, 2010.

\bibitem{Nic10a} G.~Niccoli. 
\newblock {Reconstruction of Baxter Q-operator from Sklyanin SOV for cyclic
  representations of integrable quantum models}. 
\newblock {\em Nucl. Phys.
B}, 835:263--283, 2010.

\bibitem{Nic11} G.~Niccoli. \newblock Completeness of {B}ethe {A}nsatz by {S}%
klyanin {SOV} for cyclic representations of integrable quantum models. %
\newblock {\em JHEP}, 03:123, 2011.

\bibitem{FraGSW11} H.~Frahm, J.~H. Grelik, A.~Seel, and T.~Wirth. \newblock %
Functional {B}ethe ansatz methods for the open {XXX} chain. 
\newblock {\em
J. Phys A}, 44:015001, 2011.

\bibitem{GroMN12} N.~Grosjean, J.~M. Maillet, and G.~Niccoli. \newblock On
the form factors of local operators in the lattice sine-{G}ordon model. %
\newblock {\em J. Stat. Mech.}, 2012:P10006, 2012.

\bibitem{GroN12} N.~Grosjean and G.~Niccoli. 
\newblock {The $\tau_2$-model and the chiral Potts model revisited:
  completeness of Bethe equations from Sklyanin's SOV method}. \newblock 
\emph{J. Stat. Mech.}, 2012:P11005, 2012.

\bibitem{Nic12} G.~Niccoli. \newblock Non-diagonal open spin-1/2 {XXZ}
quantum chains by separation of variables: Complete spectrum and matrix
elements of some quasi-local operators. \newblock {\em J. Stat. Mech.},
2012:P10025, 2012.

\bibitem{Nic13} G.~Niccoli. \newblock Antiperiodic spin-1/2 {XXZ} quantum
chains by separation of variables: Complete spectrum and form factors. %
\newblock {\em Nucl. Phys. B}, 870:397--420, 2013.

\bibitem{Nic13a} G.~Niccoli. \newblock An antiperiodic dynamical six-vertex
model: {I}. {C}omplete spectrum by {SOV}, matrix elements of the identity on
separate states and connections to the periodic eight-vertex model. %
\newblock {\em J. Phys. A}, 46:075003, 2013.

\bibitem{Nic13b} G.~Niccoli. \newblock Form factors and complete spectrum of 
{XXX} antiperiodic higher spin chains by quantum separation of variables. %
\newblock {\em J. Math. Phys.}, 54:053516, 2013.

\bibitem{GroMN14} N.~Grosjean, J.~M. Maillet, and G.~Niccoli. 
\newblock {On the form factors of local operators in the Bazhanov-Stroganov and
  chiral Potts models}. \newblock {\em Annales Henri Poincar{\'e}}, 16:1103,
2015.

\bibitem{FalN14} S.~Faldella and G.~Niccoli. \newblock {SOV} approach for
integrable quantum models associated with general representations on
spin-1/2 chains of the 8-vertex reflection algebra. 
\newblock {\em J. Phys.
A}, 47:115202, 2014.

\bibitem{FalKN14} S.~Faldella, N.~Kitanine, and G.~Niccoli. \newblock %
Complete spectrum and scalar products for the open spin-1/2 {XXZ} quantum
chains with non-diagonal boundary terms. \newblock {\em J. Stat. Mech.},
2014:P01011, 2014.

\bibitem{KitMN14} N.~Kitanine, J.~M. Maillet, and G.~Niccoli. \newblock Open
spin chains with generic integrable boundaries: Baxter equation and {B}ethe
ansatz completeness from separation of variables. 
\newblock {\em J. Stat.
Mech.}, 2015:P05015, 2014.

\bibitem{NicT15} G.~Niccoli and V.~Terras. \newblock Antiperiodic {XXZ}
chains with arbitrary spins: Complete eigenstate construction by functional
equations in separation of variables. \newblock {\em Lett. Math. Phys.},
105:989--1031, 2015.

\bibitem{LevNT15} D.~Levy-Bencheton, G.~Niccoli, and V.~Terras. \newblock %
Antiperiodic dynamical 6-vertex model by separation of variables {II}:
Functional equation and form factors. \newblock {\em J. Stat. Mech.}
2016:033110.

\bibitem{NicT16} G.~Niccoli and V.~Terras. \newblock The 8-vertex model with
quasi-periodic boundary conditions. \newblock {\em J. Phys. A}%
, 49:044001, 2016.

\bibitem{KitMNT16} N.~Kitanine, J.~M. Maillet, G.~Niccoli, and V.~Terras. %
\newblock On determinant representations of scalar products and form factors
in the {SoV} approach: the {XXX} case. 
\newblock {\em J. Phys. A}, 49:104002, 2016.

\bibitem{MarS16} D.~Martin and F.~Smirnov. 
\newblock {Problems with Using Separated Variables for Computing Expectation
  Values for Higher Ranks}. \newblock {\em Lett. Math. Phys.}, 106:469--484,
2016.

\bibitem{KitMNT17} N.~Kitanine, J.~M. Maillet, G.~Niccoli, and V.~Terras. %
\newblock The open {XXX} spin chain in the sov framework: scalar product of
separate states. \newblock {\em J. Phys. A}, 50(22):224001,
2017.

\bibitem{MaiNP17} J.~M. Maillet, G.~Niccoli, and B.~Pezelier. 
\newblock {Transfer matrix spectrum for cyclic representations of the 6-vertex
  reflection algebra I}. \newblock {\em SciPost Phys.}, 2:009, 2017.

\bibitem{GroLMS17} N.~Gromov, F.~Levkovich-Maslyuk, and G.~Sizov. \newblock %
New construction of eigenstates and separation of variables for {SU}({N})
quantum spin chains. \newblock {\em JHEP}, 09:111, 2017.

\bibitem{DerV18} S. E.~Derkachov and P. A.~Valinevich. \newblock Separation
of variables for the quantum {SL}(3,{C}) spin magnet: eigenfunctions of {S}%
klyanin {B}-operator. \newblock  arXiv:1807.00302 [math-ph].

\bibitem{RyaV18} P.~Ryan, D.~Volin \newblock Separated variables and wave
functions for rational gl(N) spin chains in the companion twist frame %
\newblock arXiv:1810.10996.


\bibitem{Bax73-tot}
{R. J. Baxter}.
\newblock {Eight-vertex model in lattice statistics and one-dimensional
  anisotropic Heisenberg chain. I, II, III}.
\newblock {\em {Ann. Phys.}}, {76}:{1--24 ; 25--47 ; 48--71}, {1973}.

\bibitem{Bax73}
{R. J. Baxter}.
\newblock {Spontaneous staggered polarization of the F-model}.
\newblock {\em {J. Stat. Phys.}}, {9(2)}:{145–182}, {1973}.

\bibitem{Bax76}
{R. J. Baxter}.
\newblock {Corner transfer matrices of the eight-vertex model I}.
\newblock {\em {J. Stat. Phys.}}, {15}:{485–503}, {1976}.

\bibitem{Bax77}
{R. J. Baxter}.
\newblock {Corner transfer matrices of the eight-vertex model II}.
\newblock {\em {J. Stat. Phys.}}, {76}:{1--14}, {1977}.

\bibitem{Bax78}
{R. J. Baxter}.
\newblock {Solvable eight-vertex model on an arbitrary planar lattice}.
\newblock {\em {Phil.Trans. Roy. Soc.  London}},
  {289}:{315--346}, {1978}.

\bibitem{PasG92}
{V. Pasquier and H. Saleur}.
\newblock {Common structures between finite systems and conformal field
  theories through quantum groups}.
\newblock {\em {Nucl. Phys. B}}, {330}:{523}, {1990}.

\bibitem{BatBOY95}
{M. T. Batchelor and R. J. Baxter and M. J. O'Rourke and C. M. Yung}.
\newblock {Exact solution and interfacial tension of the six-vertex model with
  anti-periodic boundary conditions}.
\newblock {\em {J.  Phys. A}}, {28}:{2759},
  {1995}.

\bibitem{YunB95}
C.~M. Yung and M.~T. Batchelor.
\newblock Exact solution for the spin-s {XXZ} quantum chain with non-diagonal
  twists.
\newblock {\em Nucl. Phys. B}, 446:461--484, 1995.

\bibitem{BazLZ97}
V.~V. Bazhanov, S.~L. Lukyanov, and A.~B. Zamolodchikov.
\newblock Integrable structure of conformal field theory {II}. {Q}-operator and
  {DDV} equation.
\newblock {\em Comm. Math. Phys.}, 190:247--278, 1997.
\newblock hep-th/9604044.

\bibitem{AntF97}
A.~Antonov and B.~Feigin.
\newblock Quantum group representation and {B}axter equation.
\newblock {\em Phys. Lett. B}, 392:115--122, 1997.
\newblock hep-th/9603105.

\bibitem{BazLZ99}
V.~V. Bazhanov, S.~L. Lukyanov, and A.~B. Zamolodchikov.
\newblock Integrable structure of conformal field theory. {III}: The
  {Y}ang-{B}axter relation.
\newblock {\em Comm. Math. Phys.}, 200:297--324, 1999.
\newblock hep-th/9805008.

\bibitem{Der99}
S.~E. Derkachov.
\newblock Baxter's {Q}-operator for the homogeneous {XXX} spin chain.
\newblock {\em J. Phys. A}, 32:5299--5316, 1999.
\newblock solv-int/9902015.

\bibitem{Pro00}
G.~P. Pronko.
\newblock On the {B}axter's {Q} operator for the {XXX} spin chain.
\newblock {\em Comm. Math. Phys.}, 212:687--701, 2000.
\newblock hep-th/9908179.

\bibitem{Kor06}
C.~Korff.
\newblock A {Q}-operator for the twisted {XXX} model.
\newblock {\em J. Phys. A}, 39:3203--3219, 2006.
\newblock math-ph/0511022.

\bibitem{Der08}
S.~E. Derkachov.
\newblock Factorization of {R}-matrix and {B}axter's {Q}-operator.
\newblock {\em J. Math. Sci.}, 151:2848, 2008.
\newblock math/0507252.

\bibitem{BazLMS10}
V.~V. Bazhanov, T.~Lukowski, C.~Meneghelli, and M.~Staudacher.
\newblock A shortcut to the {Q}-operator.
\newblock {\em J. Stat. Mech.}, page P11002, 2010.

\bibitem{Man14}
V.~V. Mangazeev.
\newblock Q-operators in the six-vertex model.
\newblock {\em Nucl. Phys. B}, 886:166--184, 2014.

\bibitem{MT15}
{C. Meneghelli and J. Teschner}.
\newblock {Integrable light-cone lattice discretizations from the universal
  R-matrix}.
\newblock {\em { Adv. Theor. Math. Phys.}}, {21}:{1189--1371}, {2017}.

\bibitem{Au-YP06} {J. H.H. Perk and H. Au-Yang}. 
\newblock {Yang-Baxter
Equations}. 
\newblock {\em {Encyclopedia of Mathematical Physics, Vol. 5,
(Elsevier Science, Oxford, 2006), pp. 465-473}.}


\end{thebibliography}
\end{document}